\documentclass[a4paper]{article}
\usepackage{amsmath}
\usepackage{amssymb}
\usepackage{amsthm}
\usepackage[numbers]{natbib}

\usepackage{natbib} 
\newtheorem{thm}{Theorem}
\newtheorem{prop}{Proposition}

\newtheorem{cor}{Corollary}
\newtheorem{lemma}{Lemma}
\newcommand{\bbr}{{\mathbb R}}
\newcommand{\bbs}{{\mathbb S}}

\usepackage{authblk}

\begin{document}
\title{Future global existence and asymptotic behaviour of solutions to the Einstein-Boltzmann system with Bianchi I symmetry}

\author[1]{Ho Lee\footnote{holee@khu.ac.kr}}
\author[2,3]{Ernesto Nungesser\footnote{ernesto.nungesser@icmat.es}}

\affil[1]{Department of Mathematics and Research Institute for Basic Science, Kyung Hee University, Seoul, 130-701, Republic of Korea}
\affil[2]{School of Mathematics, Trinity College Dublin, Dublin 2, Ireland}
\affil[3]{Instituto de Ciencias Matem\'{a}ticas, Consejo Superior de Investigaciones Cient\'{i}ficas, 28049 Madrid, Spain}

\maketitle

\begin{abstract}
In this paper we study the Einstein-Boltzmann system with Bianchi I symmetry. We show that for small initial data the corresponding solutions of the Einstein-Boltzmann system are future geodesically complete and that they isotropize and have a dust-like behaviour at late times. Detailed information about the metric and the matter terms is obtained, and the results show that the solutions tend asymptotically to the Einstein-de Sitter solution.
\end{abstract}

\section{Introduction}
In theoretical cosmology one wants to understand the future and the past behaviour of different cosmological models. This is in general too ambitious: instead, one looks at simplified models, such as the important class of homogeneous spacetimes. Before obtaining a global and detailed picture of the future and past behaviour of a cosmological model, one hopes to obtain an understanding of the asymptotic behaviour, i.e., what happened close to the initial singularity (the big bang) and what is the distant future of our Universe. Both problems are still far from being trivial. Indeed, only very special models are understood to date, and the predictions are very different from each other. In this paper we will only deal with the future direction.

How do the asymptotics depend on the way the matter content of the Universe is described? It is standard to model the Universe with a perfect fluid, but this is only one of the several options for the choice of the matter model that exist, and it is of interest to know if and how the dynamics change with a different matter model. The most common ones apart from the perfect fluid are elastic matter, field-theoretic models and kinetic models. We are interested in kinetic models and will focus in the present paper on the case where no cosmological constant is present. It is mathematically more challenging and we plan to investigate the cosmological constant case in the future. In the last years there has been a lot of progress concerning the dynamics in the Einstein-Vlasov case. For an overview of different results we refer to \cite{CH}. There the dynamics of collisionless matter is analyzed for Bianchi spacetimes which are locally rotationally symmetric (LRS). In absence of LRS symmetry, the asymptotic behaviour in the case of homogeneous spacetimes has been studied for a large class of solutions (cf.\ \cite{EN,E4}). On the other hand, not so much has been done in the Einstein-Boltzmann case. One of the reasons is that the formulation of the Einstein-Boltzmann system is not trivial in general relativity (cf.\ \cite{BC,LeeRen}). Instead, one often considers certain symmetries on the spacetime to simplify the full system of equations, and can find several results \cite{AN,Lee2,LeeRen2,ND,NDT,NT} for Bianchi type I case. In the present paper we also consider the Bianchi type I spacetime, and the objective of this paper is twofold. First, we study non-diagonal Bianchi I spacetime. In the results mentioned above the authors only deal with the diagonal case, for instance FLRW or LRS cases, but our result is valid for spacetimes of Bianchi type I which are not necessarily diagonal. Second, we obtain asymptotic behaviour of solutions at late times. For instance in \cite{ND}, the authors considered large initial data and only obtained global existence of solutions, but in this paper we consider small solutions and obtain late time behaviour.

This paper is organized as follows. In Section \ref{Sec_EB}, we derive the Einstein-Boltzmann system with Bianchi symmetry. We first review the Einstein-Vlasov system with Bianchi symmetry, and then parametrize the mass shell with an orthonormal frame to obtain the Einstein-Boltzmann system with Bianchi symmetry.  In Section \ref{Sec_EB1} we present the equations we will work with, namely the Einstein-Boltzmann system with Bianchi I symmetry. In Section \ref{Sec_Estimates}, we collect several elementary lemmas, which are used to estimate the collision operator of the Boltzmann equation. Global-in-time existence of solutions to the Einstein-Boltzmann system will be proved by a standard iteration method. In Section \ref{Sec_Einstein}, we assume that a distribution function is given, and then show that classical solutions to the Einstein equations exist globally in time. In Section \ref{Sec_Boltzmann}, we assume that a metric is given, and then show that classical solutions of the Boltzmann equation exist globally in time. We combine the results of the previous sections to obtain the main results. Future geodesic completeness and the asymptotic behaviour of solutions to the Einstein-Boltzmann system with Bianchi I symmetry are obtained under suitable small data assumptions. Details on the relativistic Boltzmann equation in the Minkowski case are collected in the appendix.

\section{The Einstein-Boltzmann system with Bianchi symmetry}\label{Sec_EB}
In this section we derive the Einstein-Boltzmann system with Bianchi symmetry. We introduce the Einstein-Vlasov system in the first subsection. In the second subsection we consider this system with Bianchi symmetry. Then parametrizing the mass shell with an orthonormal frame we obtain the Boltzmann equation with Bianchi  symmetry in the third subsection. In the last subsection we present two different parametrizations of the post-collision momenta which will be necessary in the following.

\subsection{The Einstein-Vlasov system}
A cosmological model is described via a time-orientable Lorentzian metric $g_{\alpha\beta}$
 (we will use signature $(-,+,+,+)$) on a manifold $\mathcal{M}$, which means that at each point of $\mathcal{M}$ the two halves of the light cone can be labelled past and future in a way which varies continuously from point to point. This enables to distinguish between future-pointing and past-pointing timelike vectors. Using units such that the gravitational constant and the speed of light in vacuum are equal to one, the interaction between the geometry and the matter is described by the Einstein field equations:
\begin{eqnarray*}
G_{\alpha\beta}  = 8\pi T_{\alpha \beta},
\end{eqnarray*}
where $G_{\alpha\beta}$ is the Einstein tensor and $T_{\alpha \beta}$ is the energy-momentum tensor.  Indices are lowered and raised by the spacetime metric $g_{\alpha\beta}$ and its inverse $g^{\alpha\beta}$. We use Greek indices for indices which run from 0 to 3 and Latin indices from 1 to 3. The energy-momentum tensor describes the matter content, and we will use a kinetic theory to describe it. We refer to \cite{GLW,Hans} for an introduction to the Einstein-Vlasov system.

Consider now a particle which moves only under the influence of the gravitational field. The worldline $x^\alpha$ of a particle is a null or timelike geodesic depending on whether it has zero or non-zero rest mass. The tangent vector to this curve is the 4-velocity $v^{\alpha}$ and $p_{\alpha}$ is the 4-momentum of the particle. Let $T_x \mathcal{M}$ be the tangent space and $T_x^* \mathcal{M}$ the cotangent space at a point $x^{\alpha}$ in the spacetime $\mathcal{M}$. For particles of the same type and with the same (constant) rest mass $m$, the mass shell relation is given by:
\begin{eqnarray*}
 p_{\alpha} p_{\beta}g^{\alpha\beta}+m^2=0,
\end{eqnarray*}
where the Einstein summation convention that repeated indices are to be summed over is used.
The component $p_0$ can be expressed in terms of the other components via this relation:
\begin{eqnarray*}
p_0=\frac{1}{g^{00}}\bigg(-p_ag^{0a}\pm\sqrt{(p_ag^{0a})^2-g^{00}(p_a p_b g^{ab}+m^2)}\bigg).
\end{eqnarray*}
We are interested in future pointing momenta, i.e.,
\begin{eqnarray*}
p^{0}=p_0g^{00}+p_ag^{0a}=+\sqrt{(p_a g^{0a})^2-g^{00}(p_ap_bg^{ab}+m^2)}>0.
\end{eqnarray*}
Define the \textit{phase space $P_m$ for particles of mass $m$}:
\begin{eqnarray*}
P_m=\{(x^{\alpha},p_{\alpha}):\ x^{\alpha} \in \mathcal{M},\ p_{\alpha} \in T^*_x \mathcal{M},\ p_{\alpha} p_{\beta}g^{\alpha\beta}+m^2=0,\ p^{0}>0\},
\end{eqnarray*}
which is a subset of the cotangent bundle $T^*\mathcal{M}=\{(x^{\alpha},p_{\alpha}):\ x^{\alpha} \in \mathcal{M},\ p_{\alpha} \in T^*_x \mathcal{M}\}$. Usually the geodesic equations are expressed in terms of $v^{\alpha}$, but one can do the same with $v_{\alpha}$. Using the convention that $\nabla_{E_{\beta}} E_{\alpha}=\Gamma^{\gamma}_{\alpha\beta} E_{\gamma}$, %((1.12) of \cite{WE}) 
$\nabla$ being the covariant derivative
and  $\Gamma_{\alpha\beta\mu}=g_{\alpha\delta}\Gamma^{\delta}_{\beta\mu}$ 
 we have that
\begin{eqnarray*}
0=\nabla_{ v^\alpha e_\alpha} (v^\beta e_\beta)=\left( v^\beta v^\alpha \Gamma^\gamma_{\beta\alpha} + v^\alpha \frac{\partial v^\gamma}{\partial x^\alpha}\right)  e_\gamma
=\left(  -v^\alpha v^\beta  \Gamma_{\beta\delta\alpha}  + v^\alpha \frac{\partial v_\delta }{\partial x^\alpha}\right)e_\gamma g^{\gamma\delta}.
\end{eqnarray*}
The geodesic equations are then:
\begin{eqnarray*}
\frac{dx^{\alpha}}{d\tau}=v^{\alpha}; \quad \frac{dv_{\beta}}{d\tau}=\Gamma_{\alpha\beta\mu} v^{\alpha} v^{\mu},
\end{eqnarray*}
where $\tau$ is an affine parameter.  The components of the metric connection $\Gamma_{\alpha\beta\mu}$ can be expressed via the Koszul formula in the vector basis $E_\alpha$ with commutators $[E_{\alpha},E_{\beta}]=\gamma^{\mu}_{\alpha \beta}E_{\mu}$, where $\gamma^{\mu}_{\alpha \beta}$ are called commutation functions: %, as ((1.14) of \cite{WE}):
\begin{eqnarray}\label{con}
\Gamma_{\alpha\beta\mu}=\frac12\Big(E_{\beta}(g_{\alpha\mu})+E_{\mu}(g_{\beta\alpha})-E_{\alpha}(g_{\mu\beta})
+\gamma^\delta_{\mu\beta}g_{\alpha\delta}+\gamma^{\delta}_{\alpha\mu}g_{\beta\delta}
-\gamma^{\delta}_{\beta\alpha}g_{\mu\delta}\Big),
\end{eqnarray}
Using the fact that we are summing over $v^\alpha v^\mu$ the geodesic equation for $v_\beta$ can thus be expressed as
\begin{eqnarray*}
\frac{dx^{\alpha}}{d\tau}=v^{\alpha}; \quad \frac{dv_{\beta}}{d\tau}=\frac12 \left [E_{\beta}(g_{\alpha\mu}) +2\gamma^\delta_{\mu\beta}g_{\alpha\delta}\right] v^{\alpha} v^{\mu},
\end{eqnarray*}

The restriction of the Liouville operator to the mass shell is defined as:
\begin{eqnarray*}
L=\frac{dx^{\alpha}}{ds}\frac{\partial}{\partial x^{\alpha}}+\frac{dp_b}{ds}\frac{\partial}{\partial p_{b}}.
\end{eqnarray*}
Using the geodesic equations it has the following form
\begin{eqnarray*}
L=v^{\alpha}\frac{\partial}{\partial x^{\alpha}}+m\Gamma_{\alpha b\mu} v^{\alpha} v^{\mu}\frac{\partial}{\partial p_{b}}.
\end{eqnarray*}
This operator is sometimes also called geodesic spray. The matter is described by a nonnegative real valued distribution function $f(x^\alpha,p_\alpha)$ on $P_m$. In the Vlasov case we have that
\begin{eqnarray*}
L(f)=0.
\end{eqnarray*}
It remains to define the energy-momentum tensor $T_{\alpha\beta}$ in terms of the distribution and the metric. Before that we need a Lorentz invariant volume element on the mass shell. The cotangent space has the  Lorentz invariant volume element $[-\det g^{(4)}]^{-\frac{1}{2}} dp_0 dp_1 dp_2 dp_3$, where $\det g^{(4)}$ is the determinant of the spacetime metric. Now considering $p_0$ as a dependent variable the induced Riemannian volume of the mass shell considered as a hypersurface in the cotangent space at that point is
\begin{eqnarray*}
\varpi=\frac{m}{p^0} [-\det g^{(4)}]^{-\frac{1}{2}} dp_*,
\end{eqnarray*}
where $dp_*=dp_1dp_2dp_3$. Now we define the energy momentum tensor as follows:
\begin{eqnarray*}
T_{\alpha\beta}=\int f(x^{\alpha},p_{*}) p_{\alpha}p_{\beta}\varpi.
\end{eqnarray*}
One can show that $T_{\alpha\beta}$ is divergence-free and thus it is compatible with the Einstein equations. For kinetic matter all the energy conditions hold. In particular the dominant energy condition is equivalent to the statement that in any orthonormal basis the energy density dominates the other components of $T_{\alpha\beta}$, i.e., $T_{\alpha\beta}\leq T_{00}$ for each $\alpha,\beta$. Using the mass shell relation one can see that this holds for kinetic matter. The nonnegative sum pressures condition is in our case equivalent to $g_{ab}T^{ab} \ge 0$. Another important quantity is the particle current density, which is defined as
\begin{eqnarray*}
N^{\alpha}=\int  f(x^{\alpha},p_{*}) p^{\alpha}\varpi.
\end{eqnarray*}
It holds that
\begin{eqnarray*}
\nabla_{\alpha} N^{\alpha}=E_{\alpha} N^{\alpha}+\Gamma^{\alpha}_{\mu\alpha} N^{\mu}=0,
\end{eqnarray*}
which is an expression of the conservation of the number of particles. For simplicity we will assume that our system of particles is made of a single species of particles (no mixtures) and that all particles have unit mass, i.e., the units are chosen such that $m=1$, so that a distinction between velocities and momenta is not necessary and the mass shell relation is
\begin{eqnarray}
p_{\alpha}p_{\beta}g^{\alpha\beta}=-m^2=-1.\label{mass_shell}
\end{eqnarray}

\subsection{The 3+1 decomposition of the Einstein-Vlasov system with Bianchi symmetry}
A Bianchi spacetime is defined to be a spatially homogeneous spacetime whose isometry group possesses a three-dimensional subgroup that acts simply transitively on the spacelike orbits. A Bianchi spacetime admits a Lie algebra of Killing vector fields. These vector fields are tangent to the group orbits, which are the surfaces of homogeneity. Using a left-invariant frame, the metric induced on the spacelike hypersurfaces depends only on the time variable. If $\omega^a$ denote the one-forms dual to the spatial frame vectors, the metric of a Bianchi spacetime in the left-invariant frame is written as
\begin{eqnarray}\label{metric}
g^{(4)}=-dt^2+g,\quad g=g_{ab}(t) \omega^a(t) \omega^b(t).
\end{eqnarray}
Due to the simple form of the metric in the Bianchi case, it is useful to do the 3+1 decomposition of the Einstein equations in the left-invariant frame. Using the sign conventions as in \cite{RA} we obtain with zero shift and the lapse function equal to one ($\beta=0$ and $\alpha=1$ in (2.29) and (2.34) of \cite{RA}) the following evolution equations:
\begin{eqnarray*}
\dot{g}_{ab}&=&-2k_{ab},\\
\dot{k}_{ab}&=&R_{ab}+k\, k_{ab}-2(g^{cd}k_{bd})k_{ac}-8\pi S_{ab}+4\pi g_{ab}(S-\rho),
\end{eqnarray*}
where $k_{ab}$ is the second fundamental form, $k=k_{ab}g^{ab}$ its trace, $R_{ab}$ the Ricci tensor of the spatial metric $g_{ab}$, $S=g^{ab}S_{ab}$, and the dot denotes the derivative with respect to time $t$. The constraint equations are given as ((2.26) and (2.27) of \cite{RA})
\begin{eqnarray*}
R-k^{ab}k_{ab}+k^2&=&16\pi \rho,\\
\nabla^a k_{ab}&=&-8\pi T_{0b},
\end{eqnarray*}
where $R=R_{ab}g^{ab}$ is the scalar curvature. Since we use a left-invariant frame the distribution function $f$ will not depend on $x^a$. Moreover since $g_{00}=g^{00}=-1$ and $g^{0a}=0$ we have that $p^0=-p_0=\sqrt{1+p_a p_b g^{ab}}$. Defining $\rho=T_{00}$, $J_{a}=-T_{0a}$ and $S_{ij}=T_{ij}$ the frame components of the energy-momentum tensor and the particle current density are thus
\begin{eqnarray*}
\rho&=&(\det g)^{-\frac12}  \int_{\bbr^3} f(t,p_*) \sqrt{1+p_ap_bg^{ab}} dp_*,\\
J_{i}&=&(\det g)^{-\frac12} \int_{\bbr^3} f(t,p_*) p_i  dp_*,\\
S_{ij}&=&(\det g)^{-\frac12} \int_{\bbr^3} f(t,p_*) \frac{p_i p_j} {\sqrt{1+p_ap_bg^{ab}}}dp_*,\\
N^{0}&=& (\det g)^{-\frac12} \int_{\bbr^3} f(t,p_*) dp_*,\\
N^{i}&=& (\det g)^{-\frac12} \int_{\bbr^3} f(t,p_*) \frac{p^i} {\sqrt{1+p_ap_bg^{ab}}}dp_*.
\end{eqnarray*}
The conservation of the number of particles implies then
\begin{eqnarray}\label{conspart}
\dot{N}^{0}+{\Gamma^{\alpha}}_{\beta\alpha}N^{\beta}=0.
\end{eqnarray}

We now consider the Vlasov equation. Since $b$ is spacelike
\begin{eqnarray*}
\Gamma_{\alpha b\mu} p^{\alpha}p^{\mu}=\gamma^\delta_{\mu b}g_{\alpha\delta}p^{\alpha}p^{\mu}.
\end{eqnarray*}
Consider now two different cases.
\begin{itemize}
\item In the metric approach, a basis of the frame vectors is chosen, such that they are \textbf{time-independent} and the commutation functions are made equal to the structure constants of the Lie algebra which are denoted by $C^i_{jk}$. Only the spatial components remain and as a result we have
\begin{eqnarray*}
\frac{\partial f}{\partial t}+\frac{1}{p^0}C^d_{ab}p^{a}p_{d}\frac{\partial f}{\partial p_b}=0.
\end{eqnarray*}
\item In the orthonormal frame approach, $g_{ab}=\eta_{ab}$. On the other hand the commutation functions and the spacelike frame vectors depend on time. However $\gamma^0_{\alpha\beta}=0$ and $\gamma^a_{0b}=\hat{k}^{a}_{b}$, where we use a hat to indicate that we are in the orthonormal frame. As a result
\begin{eqnarray*}
\frac{\partial \hat{f}}{\partial t}+ \Big(\hat{k}^d_b \hat{p}_d
+(\hat{p}^0)^{-1} \gamma^d_{ab}\hat{p}_d\hat{p}^{a}\Big)\frac{\partial \hat{f}}{\partial \hat{p}_b}=0.
\end{eqnarray*}
\end{itemize}
For more details for these approaches we refer to \cite{WE}. Now by coupling the Vlasov equation to the Einstein equations through the energy-momentum tensor, we obtain the Einstein-Vlasov system with Bianchi symmetry.

\subsection{The Boltzmann equation with Bianchi symmetry}
We now have to introduce the collision operator. The idea is to use first the orthonormal frame, because then the collision operator has the same form as for the Minkowski case. Then, we have
\begin{eqnarray*}
\frac{\partial \hat{f}}{\partial t}+ \Big(\hat{k}^d_b \hat{p}_d
+(\hat{p}^0)^{-1} \gamma^d_{ab}\hat{p}_d\hat{p}^{a}\Big)\frac{\partial \hat{f}}{\partial \hat{p}_b}=Q(\hat{f},\hat{f}),
\end{eqnarray*}
where $Q(\hat{f},\hat{f})$ is the collision operator as in the Minkowski case. If we make the transformations from one frame to the other we arrive at
\begin{eqnarray*}
\frac{\partial f}{\partial t}+\frac{1}{p^0}C^d_{ab}p^{a}p_{d}\frac{\partial f}{\partial p_b}=Q(f,f),
\end{eqnarray*}
where $Q(f,f)$ is now the collision operator for the Bianchi case and $f$ depends on $p_*$. The details for the collision operator in the Minkowski case are presented in the appendix. Let us introduce now some standard terminology. If two particles with momenta $p^{\alpha}$ and $q^{\alpha}$ collide, due to the conservation of energy and momentum the following relation holds between the momenta $(p^{\alpha}, q^{\alpha})$ and the post-collision momenta $(p'^{\alpha},q'^{\alpha})$:
\begin{eqnarray*}
p'^{\alpha}+q'^{\alpha}=p^{\alpha}+q^{\alpha}.
\end{eqnarray*}
The post-collision momenta $p'^\alpha$ will be parametrized in two ways with parameters $\omega,\xi\in\bbs^2$. The relative momentum $h$ and the total energy $s$ are defined by
\begin{eqnarray*}
h=\sqrt{(p_\alpha-q_\alpha)(p^\alpha-q^\alpha)},\quad
s=-(p_\alpha+q_\alpha)(p^\alpha+q^\alpha).
\end{eqnarray*}
Using these variables the M{\o}ller velocity $v_M$ is defined as
\begin{eqnarray*}
v_M=\frac{h\sqrt{s}}{4p^0q^0}.
\end{eqnarray*}
For notational simplicity we also use
\begin{eqnarray*}
n^\alpha=p^\alpha+q^\alpha,\quad u^\alpha=p^\alpha-q^\alpha,
\end{eqnarray*}
and then we may also write $h^2=u_\alpha u^\alpha$ and $s=-n_\alpha n^\alpha$. Finally we denote by $\sigma=\sigma(h,\theta)$ the scattering kernel, where $\theta$ is the scattering angle, which is the angle between $u^\alpha$ and $u'^\alpha$. For background on the relativistic kinetic theory and intuition on the variables and expressions introduced, we refer to \cite{GLW} and pages 94--116 of \cite{Glassey}. For a more general presentation of the non-relativistic Boltzmann equation we refer to \cite{Vil}.

\subsection{Two representations of the collision operator and the post-collision momenta}
We generalize the Boltzmann collision operator in the Minkowski case to the Bianchi case. One way to do this is to use an orthonormal frame in the momentum space \cite{LeeRen}. Applying the usual Gram-Schmidt process, which also works for pseudo-Riemannian metrics as mentioned in \cite{LeeRen}, we obtain an orthonormal frame $\{e_\mu\}$ from a given frame $\{E_\alpha\}$ such that $e_\mu=e_\mu^\alpha E_\alpha$ or $E_\alpha=\theta^\mu_\alpha e_\mu$, where $\theta^\mu_\alpha$ is the inverse of $e^\alpha_\mu$. It satisfies the following relations:
\begin{eqnarray*}
p^\alpha=e_\mu^\alpha \hat{p}^\mu,\quad
p_\alpha=\theta^\mu_\alpha\hat{p}_\mu,\quad
g_{\alpha\beta}=\theta^\mu_\alpha\theta^\nu_\beta\eta_{\mu\nu},\quad
g^{\alpha\beta}=e_\mu^\alpha e_\nu^\beta\eta^{\mu\nu},\label{components}
\end{eqnarray*}
where $p_\alpha=g_{\alpha\beta}p^\beta$ and $\hat{p}_\mu=\eta_{\mu\nu}\hat{p}^\nu$. In the Bianchi case with the metric \eqref{metric}, we may take $e_0=\frac{\partial}{\partial t}$ to obtain
\begin{eqnarray*}
e^0_0=\theta^0_0=1,\quad
e^a_0=e^0_a=\theta^a_0=\theta^0_a=0,\quad a=1,2,3.
\end{eqnarray*}
Thus, $p_0=\hat{p}_0$, $p^0=\hat{p}^0$, and
\begin{eqnarray*}
p_a=\theta^b_a\hat{p}_b,\quad
p^a=e_b^a \hat{p}^b,\quad
g_{ab}=\theta^c_a\theta^d_b\eta_{cd},\quad
g^{ab}=e_c^a e_d^b\eta^{cd}.
\end{eqnarray*}
We can also compute the other components as follows:
\begin{eqnarray}
e^a_b=
\left(
\begin{array}{c}
e_1\\
e_2\\
e_3
\end{array}
\right)=
\left(
\begin{array}{ccc}
\frac{1}{\sqrt{m_1}}
& 0
& 0\\
\frac{-g_{12}}{\sqrt{m_1}\sqrt{m_2}}
& \frac{g_{11}}{\sqrt{m_1}\sqrt{m_2}}
& 0 \\
\frac{g_{12}g_{23}-g_{13}g_{22}}{\sqrt{m_2}\sqrt{m_3}}
& \frac{-g_{11}g_{23}+g_{12}g_{13}}{\sqrt{m_2}\sqrt{m_3}}
& \frac{g_{11}g_{22}-(g_{12})^2}{\sqrt{m_2}\sqrt{m_3}}
\end{array}
\right),\label{e^a_b}
\end{eqnarray}
and its inverse
\begin{eqnarray}
\theta^a_b=
\left(
\begin{array}{c}
\theta_1\\
\theta_2\\
\theta_3
\end{array}
\right)=
\left(
\begin{array}{ccc}
\frac{g_{11}}{\sqrt{m_1}}
& 0
& 0\\
\frac{g_{12}}{\sqrt{m_{1}}}
& \frac{g_{11}g_{22}-(g_{12})^2}{\sqrt{m_1}\sqrt{m_2}}
& 0 \\
\frac{g_{13}}{\sqrt{m_1}}
& \frac{g_{11}g_{23}-g_{12}g_{13}}{\sqrt{m_1}\sqrt{m_2}}
& \frac{\det g}{\sqrt{m_2}\sqrt{m_3}}
\end{array}
\right),\label{theta^a_b}
\end{eqnarray}
where $m_1$, $m_2$, and $m_3$ are the principal minors of the matrix $g_{ab}$, i.e., the quantities $g_{11}$, $g_{11}g_{22}-(g_{12})^2$, and $\det g$, respectively. In fact, the matrix $\theta^a_b$ is uniquely determined, since $g_{ab}$ is symmetric and positive definite, which is known as the Cholesky factorization. Note that in order to estimate the quantities $e_b^a$ and $\theta^a_b$ we need to estimate the lower bounds of the principal minors of the matrix $g_{ab}$. The following two lemmas give the lower bounds of the principal minors. Below, $M_n$ denotes the set of all $n\times n$ matrices over $\bbr$. For the proof of the lemmas we refer to \cite{HJ}.

\begin{lemma}{(Hadamard's inequality)}\label{Lem_Hadamard}
Let $A=(a_{ij})\in M_n$ be positive definite. Then,
\[
\det A\leq a_{11}\cdots a_{nn}
\]
with equality if and only if $A$ is diagonal.
\end{lemma}
\begin{lemma}{(Fischer's inequality)}\label{Lem_Fischer}
Suppose that the partitioned symmetric matrix
\[
H=\bigg(
\begin{array}{cc}
A & B \\
B^T & C
\end{array}
\bigg)
\in M_{p+q},\quad A\in M_p,\quad C\in M_q
\]
is positive definite. Then,
\[
\det H\leq (\det A)(\det C).
\]
\end{lemma}

The following estimates are now easily obtained,
\begin{eqnarray}\label{est_principal_minors}
\frac{1}{m_1}\leq\frac{g_{22}g_{33}}{\det g},\quad
\frac{1}{m_2}\leq\frac{g_{33}}{\det g},
\end{eqnarray}
and we can conclude that the quantities $e^a_b$ and $\theta^a_b$ are bounded as long as $(\det g)^{-1}$ and $g_{ab}$ are bounded. The above inequalities are also used to estimate the asymptotic behaviour of the components $e^a_b$ and $\theta^a_b$.

In an orthonormal frame the representations of collision operator and post-collision momentum become identical to the Minkowski case (see Appendix). Hence, we have
\begin{eqnarray*}
Q(\hat{f},\hat{f})=\int_{\bbr^3}\int_{\bbs^2}\frac{h\sqrt{s}}{4\hat{p}^0\hat{q}^0}
\sigma(h,\theta)
\Big(\hat{f}(\hat{p}')\hat{f}(\hat{q}')-\hat{f}(\hat{p})\hat{f}(\hat{q})\Big)
d\omega d\hat{q}.\label{QMS}
\end{eqnarray*}
Here, $h$ and $s$ are understood as
\begin{eqnarray*}
h=\sqrt{\eta_{\mu\nu}(\hat{p}^\mu-\hat{q}^\mu)(\hat{p}^\nu-\hat{q}^\nu)},\quad
s=-\eta_{\mu\nu}(\hat{p}^\mu+\hat{q}^\mu)(\hat{p}^\nu+\hat{q}^\nu),
\end{eqnarray*}
where $\hat{p}^\mu$ and $\hat{q}^\mu$ denote the components of momenta in the orthonormal frame. The post-collision momentum is written by
\begin{eqnarray*}
\left(
\begin{array}{c}
\hat{p}'^0\\
\hat{p}'^k
\end{array}
\right)=
\left(
\begin{array}{c}
\displaystyle
\frac{\hat{p}^0+\hat{q}^0}{2}
+\frac{h}{2}
\frac{(\hat{n}\cdot\omega)}{\sqrt{s}}\\
\displaystyle
\frac{\hat{p}^k+\hat{q}^k}{2}
+\frac{h}{2}
\bigg(\omega^k
+\bigg(\frac{\hat{n}^0}{\sqrt{s}}-1\bigg)
\frac{(\hat{n}\cdot\omega)\hat{n}^k}{|\hat{n}|^2}\bigg)
\end{array}
\right),\label{p'MS}
\end{eqnarray*}
where $\hat{n}^\mu=\hat{p}^\mu+\hat{q}^\mu$. In Glassey-Strauss's framework, the above quantities are written as follows:
\begin{eqnarray*}
Q(\hat{f},\hat{f})=\int_{\bbr^3}\int_{\bbs^2}
\frac{hs(\hat{n}^0)^2\sigma(h,\theta)}{4\hat{p}^0\hat{q}^0
((\hat{n}^0)^2-(\hat{n}\cdot\xi)^2)^{3/2}}
\Big(\hat{f}(\hat{p}')\hat{f}(\hat{q}')-\hat{f}(\hat{p})\hat{f}(\hat{q})\Big)d\xi d\hat{q},\label{QMGS}
\end{eqnarray*}
and
\begin{eqnarray*}
\left(
\begin{array}{c}
\hat{p}'^0\\
\hat{p}'^k
\end{array}
\right)=
\left(
\begin{array}{c}
\displaystyle
\frac{\hat{p}^0+\hat{q}^0}{2}+\frac{h}{2}
\frac{(\hat{n}\cdot\xi)}{\sqrt{(\hat{n}^0)^2-(\hat{n}\cdot\xi)^2}}\\
\displaystyle
\frac{\hat{p}^k+\hat{q}^k}{2}+\frac{h}{2}\frac{\hat{n}^0\xi^k}
{\sqrt{(\hat{n}^0)^2-(\hat{n}\cdot\xi)^2}}
\end{array}
\right).\label{p'MGS}
\end{eqnarray*}
We now switch to the $p_*$ and $q_*$ variables through the transformations $e^\alpha_\mu$ or $\theta^\mu_\alpha$, and then obtain representations of the collision operator and post-collision momentum in the original frame. In the end what we obtain is the collision operator $Q$ in Strain's framework,
\begin{eqnarray*}
Q(f,f)=(\det g)^{-\frac12}\int_{\bbr^3}\int_{\bbs^2}v_M
\sigma(h,\theta)\Big(f(p'_*)f(q'_*)-f(p_*)f(q_*)\Big)d\omega dq_*
\end{eqnarray*}
with the parametrization of post-collision momenta
\begin{eqnarray*}
\left(
\begin{array}{c}
p'^0\\
p'_i
\end{array}
\right)=
\left(
\begin{array}{c}
\displaystyle
\frac{p^0+q^0}{2}+\frac{h}{2}\frac{n_ae^a_b\omega^b}{\sqrt{s}}\\
\displaystyle
\frac{p_i+q_i}{2}+\frac{h}{2}\bigg(
g_{ia}e^a_b\omega^b
+\bigg(\frac{n^0}{\sqrt{s}}-1\bigg)
\frac{n_ae^a_b\omega^b n_i}{g^{cd}n_cn_d}
\bigg)
\end{array}
\right),
\end{eqnarray*}
and in Glassey-Strauss's framework,
\begin{eqnarray*}
Q(f,f)=(\det g)^{-\frac12} \int_{\bbr^3}\int_{\bbs^2}
\frac{v_M\sqrt{s}(n^0)^2\sigma(h,\theta)}{((n^0)^2-(n_ae^a_b\xi^b)^2)^{3/2}}
\Big(f(p'_*)f(q'_*)-f(p_*)f(q_*)\Big)d\xi dq_*
\end{eqnarray*}
with the parametrization of post-collision momenta
\begin{eqnarray*}
\left(
\begin{array}{c}
p'^0\\
p'_i
\end{array}
\right)=
\left(
\begin{array}{c}
\displaystyle
\frac{p^0+q^0}{2}+\frac{h}{2}
\frac{n_ae^a_b\xi^b}
{\sqrt{(n^0)^2-(n_ce^c_d\xi^d)^2}}\\
\displaystyle
\frac{p_i+q_i}{2}+\frac{h}{2}\frac{n^0g_{ia}e^a_b\xi^b}
{\sqrt{(n^0)^2-(n_ce^c_d\xi^d)^2}}
\end{array}
\right).
\end{eqnarray*}
For simplicity we used upper indices for zeroth components since they only differ in sign, i.e., $p'_0=-p'^0$, and moreover we will regard $p'^0$, as well as $p^0$, as a function of spatial components through the mass shell condition \eqref{mass_shell} as in the Vlasov case. Note that the collision operator can be written as $Q=Q_+-Q_-$ in both cases, which are called the gain and the loss terms, respectively.

\section{The Einstein-Boltzmann system with Bianchi I symmetry}\label{Sec_EB1}
The Einstein-Boltzmann system with Bianchi I symmetry is obtained by combining the equations derived so far. At the end of the section we make assumptions on the scattering kernel and the Hubble variable. Several basic estimates of the Hubble variable will also be given. Note that the Ricci tensor, the scalar curvature and the coefficients of the metric connection with spacelike indices vanish in the Bianchi I case. Now, the evolution equations of the Einstein equations are
\begin{eqnarray}
\dot{g}_{ab}&=&-2k_{ab},\label{evolution1}\\
\dot{k}_{ab}&=&(g^{cd}k_{cd})k_{ab}-2(g^{cd}k_{bd})k_{ac}-8\pi S_{ab}+4\pi g_{ab}(S-\rho),\label{evolution2}
\end{eqnarray}
and the constraint equations are
\begin{eqnarray}
-k^{ab}k_{ab}+k^2&=&16\pi \rho,\label{constraint1}\\
0&=&-8\pi T_{0b}.\label{constraint2}
\end{eqnarray}
The first expression of the Boltzmann equation is
\begin{eqnarray}
\frac{\partial f}{\partial t}
=(\det g)^{-\frac12}\int_{\bbr^3}\int_{\bbs^2}v_M
\sigma(h,\theta)\Big(f(p'_*)f(q'_*)-f(p_*)f(q_*)\Big)d\omega dq_*,\label{boltzmann1}
\end{eqnarray}
where
\begin{eqnarray}
\left(
\begin{array}{c}
p'^0\\
p'_i
\end{array}
\right)=
\left(
\begin{array}{c}
\displaystyle
\frac{p^0+q^0}{2}+\frac{h}{2}\frac{n_ae^a_b\omega^b}{\sqrt{s}}\\
\displaystyle
\frac{p_i+q_i}{2}+\frac{h}{2}\bigg(
g_{ia}e^a_b\omega^b
+\bigg(\frac{n^0}{\sqrt{s}}-1\bigg)
\frac{n_ae^a_b\omega^bn_i}{g^{cd}n_cn_d}
\bigg)
\end{array}
\right),\label{p'1}
\end{eqnarray}
and the second expression is
\begin{eqnarray}
\frac{\partial f}{\partial t}=(\det g)^{-\frac12} \int_{\bbr^3}\int_{\bbs^2}
\frac{v_M\sqrt{s}(n^0)^2\sigma(h,\theta)}{((n^0)^2-(n_ae^a_b\xi^b)^2)^{3/2}}
\Big(f(p'_*)f(q'_*)-f(p_*)f(q_*)\Big)d\xi dq_*,\label{boltzmann2}
\end{eqnarray}
where
\begin{eqnarray}
\left(
\begin{array}{c}
p'^0\\
p'_i
\end{array}
\right)=
\left(
\begin{array}{c}
\displaystyle
\frac{p^0+q^0}{2}+\frac{h}{2}
\frac{n_ae^a_b\xi^b}
{\sqrt{(n^0)^2-(n_ce^c_d\xi^d)^2}}\\
\displaystyle
\frac{p_i+q_i}{2}+\frac{h}{2}\frac{n^0g_{ia}e^a_b\xi^b}
{\sqrt{(n^0)^2-(n_ce^c_d\xi^d)^2}}
\end{array}
\right).\label{p'2}
\end{eqnarray}
They are coupled to each other through the energy-momentum tensor
\begin{eqnarray}
\rho&=&(\det g)^{-\frac{1}{2}} \int_{\bbr^3} f(t,p_*)(1+g^{cd}p_{c}p_d)^{\frac{1}{2}}dp_*,\label{rho}\\
S_{ab}&=&(\det g)^{-\frac{1}{2}} \int_{\bbr^3} f(t,p_*)p_a p_b(1+g^{cd}p_{c}p_d)^{-\frac{1}{2}}dp_*.\label{S_ab}
\end{eqnarray}
In this paper the Einstein-Boltzmann system with Bianchi I symmetry will refer to the system of equations \eqref{evolution1}--\eqref{S_ab}, and global-in-time existence and asymptotic behaviour of solutions will be studied. We expect the spacetime to behave as the Einstein-de Sitter model at late times. Hence, it will be useful to use the following notation throughout the text for a scaled version of the spatial metric and its inverse:
\begin{eqnarray}\label{bar}
g_{ab}=t^{\frac43} \bar{g}_{ab}\quad\mbox{and}\quad g^{ab}=t^{-\frac43} \bar{g}^{ab}.
\end{eqnarray}
Throughout the paper, $e^a_b$ and $\theta^a_b$ will denote the components of an orthonormal frame and its inverse, respectively. Change of variables between $p$ and $\hat{p}$ will be used frequently. Indices will be lowered by $g_{ab}$ for $p^a$, i.e., $p_a=g_{ab}p^b$, but will be lowered by $\eta_{ab}$ in the orthonormal frame case, i.e., $\hat{p}_a=\eta_{ab}\hat{p}^b$.

In the rest of the section we make two assumptions on the Einstein-Boltzmann system. The first one is on the scattering kernel of the Boltzmann equation, and the second one is on the Hubble variable which will be defined below. Several basic estimates of the Hubble variable will also be given, and the estimates hold for any classical solutions. Those estimates will be improved in later sections by assuming smallness of initial data.

\subsection{Assumption on the scattering kernel}
The scattering kernel is a quantity which depends on the relative momentum $h$ and the scattering angle $\theta$. The relative momentum $h$ is defined as $h^2=u_\alpha u^\alpha$, and the scattering angle $\cos\theta=u_\alpha u'^\alpha/h^2$, where we used $h^2=u_\alpha u^\alpha=u'_\alpha u'^\alpha$. For a technical reason, in the present paper, we will assume that the scattering kernel is given by $\sigma=\sigma(p_*,q_*,\omega)$. To be precise, we assume that the scattering kernel, when written in an orthonormal frame, i.e.,
\begin{eqnarray*}
\sigma=\hat{\sigma}(\hat{p},\hat{q},\omega),
\end{eqnarray*}
satisfies the following:
\begin{itemize}
\item It is non-negative and bounded.
\item It is differentiable with respect to $\hat{p}$, and its derivatives are bounded.
\item It is supported in $\Big\{|\hat{u}|^2\leq |\hat{u}'|^2+ Ct^{-\frac43}\Big\}$ for some positive number $C$.
\end{itemize}
The third assumption on the scattering kernel may not seem to be physically well-motivated. However, we remark that for any covariant variables $p_*$ and $q_*$ the support will eventually cover all $\omega\in\bbs^2$. By the estimate \eqref{e2} we have $|\hat{u}|^2-|\hat{u}'|^2\leq h^2(p^0q^0-1)$ and note that
\begin{eqnarray*}
p^0=\sqrt{1+g^{ab}p_ap_b},\quad h^2=-(p^0-q^0)^2+g^{ab}(p_a-q_a)(p_b-q_b).
\end{eqnarray*}
In the present paper we will show that $g^{ab}\sim t^{-\frac43}$. This implies that the non-negative quantity $h^2(p^0q^0-1)$ decays to zero with a rate $t^{-\frac83}$. Hence, we can see that the inequality $|\hat{u}|^2\leq |\hat{u}'|^2+Ct^{-\frac43}$ will eventually hold.

\subsection{Assumption on the Hubble variable and time origin choice}
Instead of the mean curvature, i.e., the trace of the second fundamental form $k=g^{ab}k_{ab}$, in cosmology often the Hubble variable is used:
\begin{eqnarray*}
H=-\frac13 k.
\end{eqnarray*}
In this paper we assume that the Hubble variable $H$ is initially positive (the trace of the second fundamental form is initially negative), i.e.,
\begin{eqnarray*}
H(t_0)>0,
\end{eqnarray*}
which corresponds to an initially expanding universe. Then, the first constraint equation implies that if the spacetime is neither Bianchi IX nor flat, then $H>0$ for all time. The initial data set corresponds to the triple $(g_{ij}^0, k_{ij}^0, f^0)$. Since one can always add an arbitrary constant to the time origin, we choose our time coordinate and time origin such that $t_0=-2(k^0)^{-1}$, which is just some positive number. In other words we can set without loss of generality
\begin{eqnarray}
H(t_0)=\frac{2}{3t_0}.\label{choice}
\end{eqnarray}

\subsection{Basic estimates of the Hubble variable}
It is useful to express the second fundamental form in terms of its trace free part $\sigma_{ab}$ and the Hubble variable:
\begin{eqnarray}\label{tracefree}
k_{ab}=\sigma_{ab}-Hg_{ab}.
\end{eqnarray}
The evolution equation for the Hubble variable is
\begin{eqnarray}\label{H}
\dot{H}=-3H^2+\frac{4\pi}{3}(3\rho-S).
\end{eqnarray}
Let us use the notation
\begin{eqnarray*}
F=\frac{1}{4H^2}\sigma_{ab}\sigma^{ab}
\end{eqnarray*}
to denote the shear. If we substitute the energy density with the constraint equation, then we obtain $F=\frac32(1-\Omega)$ with $\Omega=\frac{8\pi \rho}{3H^2}$. The evolution equation for $H$ can then be expressed as
\begin{eqnarray}\label{H2}
\frac{d(H^{-1})}{dt}=\frac32 +F+\frac{4\pi S}{3H^2}.
\end{eqnarray}
From (\ref{H}) using the inequality $\rho\geq S$ and from (\ref{H2}) using the fact that both $F$ and $S$ are positive we obtain:
\begin{eqnarray*}
\frac32\leq \frac{d(H^{-1})}{dt}\leq 3,
\end{eqnarray*}
which integrating leads to
\begin{eqnarray}
\frac{1}{3t-\frac32 t_0}\leq H\leq \frac{2}{3}t^{-1},\label{estH}
\end{eqnarray}
where we used \eqref{choice}. We remark that the estimate \eqref{estH} holds for any classical solution, and this estimate will be improved by assuming smallness.

We collect several evolution equations which will be used to obtain detailed asymptotic behaviour in later sections. The evolution equation for $F$ is
\begin{eqnarray}\label{F}
\dot{F}=-3H\bigg(F\bigg(1-\frac23 F -\frac{8\pi S}{9H^2}\bigg)-\frac{4\pi}{3H^3}S_{ab}\sigma^{ab}\bigg).
\end{eqnarray}
For the Bianchi I case the conservation of the number of particles \eqref{conspart} is simply
\begin{eqnarray}\label{N0}
\dot{N}^0=-3HN^0.
\end{eqnarray}
Another useful relation concerns the determinant of the metric $g$ ((2.30) of \cite{RA}):
\begin{eqnarray}\label{det}
\frac{d(\log \det g)}{dt}= -2 k=6H.
\end{eqnarray}

\section{Estimates of the collision kernel}\label{Sec_Estimates}
In this section we collect several elementary lemmas concerning estimates of quantities related to the Boltzmann collision operator. The proofs are straightforward, and the lemmas will be used in later sections to estimate the distribution function of the Boltzmann equation.

\begin{lemma}
The following estimates hold:
\begin{align}
&s=4+h^2,\quad 2\leq\sqrt{s},\quad h\leq\sqrt{s}\leq 2\sqrt{p^0q^0},\label{e1}\\
&h^2\leq u_au^a\leq p^0q^0h^2,\label{e2}\\
&\frac{q^0}{p^0}+\frac{p^0}{q^0} \leq s\leq (n^0)^2-(n_ae^a_b\xi^b)^2.\label{e3}
\end{align}
\end{lemma}
\begin{proof}
It is easy to compute that
\begin{eqnarray*}
s=4+h^2,
\end{eqnarray*}
which obviously gives us that
\begin{eqnarray*}
\sqrt{s}\geq h, \quad \sqrt{s} \geq 2.
\end{eqnarray*}
For the inequalities in \eqref{e1} and \eqref{e2}, we use
\begin{align*}
(1\pm g_{ab}p^aq^b)^2&=1\pm 2g_{ab}p^aq^b+(g_{ab}p^aq^b)^2\\
&\leq 1+g_{ab}p^ap^b+g_{ab}q^aq^b+(g_{ab}p^ap^b)(g_{cd}q^cq^d)=(p^0)^2(q^0)^2.
\end{align*}
Then, the last inequality of \eqref{e1} is verified as follows:
\begin{align*}
s&=-(p_\alpha+q_\alpha)(p^\alpha+q^\alpha)=2+2p^0q^0-2g_{ab}p^aq^b\leq 4p^0q^0.
\end{align*}
Since $h^2=-(u^0)^2+u_au^a$, the first inequality of \eqref{e2} is clear, and the second inequality is proved in a similar way:
\begin{align*}
h^2&=2p^0q^0-2(1+g_{ij}p^iq^j)\\
&=2\bigg(\frac{(p^0)^2(q^0)^2-(1+g_{ab}p^aq^b)^2}{p^0q^0+1+g_{cd}p^cq^d}\bigg)
\geq\frac{(p^0)^2(q^0)^2-(1+g_{ab}p^aq^b)^2}{p^0q^0}\\
&=\frac{1+g_{ab}p^ap^b+g_{ab}q^aq^b+(g_{ab}p^ap^b)(g_{cd}q^cq^d)
-1-2g_{ab}p^aq^b-(g_{ab}p^aq^b)^2}{p^0q^0}\\
&\geq \frac{g_{ab}(p^a-q^a)(p^b-q^b)}{p^0q^0},
\end{align*}
where we used the fact that $g_{0b}=0$.

For the estimate \eqref{e3}, we use the definition of $s$ as follows:
\begin{align*}
s&=(n^0)^2-g_{ab}n^an^b\\
&=(n^0)^2-(g_{ab}n^an^b)(g_{cd}e^c_i\xi^ie^d_j\xi^j)
\leq (n^0)^2-(g_{ab}n^ae^b_c\xi^c)^2,
\end{align*}
where we used $g_{ab}e^a_ce^b_d=\eta_{cd}$, and this
gives the upper bound of $s$.
The lower bound for $s$ in \eqref{e3} is proved as
\begin{align*}
s&=2+2p^0q^0-2g_{ab}p^aq^b
\geq 2+2p^0q^0-2\sqrt{g_{ab}p^ap^b}\sqrt{g_{cd}q^cq^d}\\
&=2+2\bigg(\frac{(p^0)^2(q^0)^2-g_{ab}p^ap^bg_{cd}q^cq^d}{p^0q^0+\sqrt{g_{ab}p^ap^b}\sqrt{g_{cd}q^cq^d}}\bigg)\\
&\geq 2+\frac{1+g_{ab}p^ap^b+g_{ab}q^aq^b}{p^0q^0}
=\frac{2p^0q^0+(p^0)^2+(q^0)^2-1}{p^0q^0}\geq\frac{q^0}{p^0}+\frac{p^0}{q^0},
\end{align*}
 and this completes the proof of the lemma.
\end{proof}

\begin{lemma}\label{Lem_partial_p}
The following estimates hold:
\begin{align}
&\partial_{p_i}p^0=\frac{p^i}{p^0},\label{ed1}\\
&\partial_{p_i}h=\frac{q^0}{h}\bigg(\frac{p^i}{p^0}-\frac{q^i}{q^0}\bigg),\label{ed2}\\
&\partial_{p_i}\sqrt{s}=\frac{q^0}{\sqrt{s}}
\bigg(\frac{p^i}{p^0}-\frac{q^i}{q^0}\bigg),\label{ed3}\allowdisplaybreaks\\
&\partial_{p_i}\sqrt{(n^0)^2-(n_ae^a_b\xi^b)^2}=\frac{1}{\sqrt{(n^0)^2-(n_ae^a_b\xi^b)^2}}\nonumber\\
&\quad\hspace{2cm}\times\bigg(q^0\bigg(\frac{p^i}{p^0}-\frac{q_ae^a_b\xi^be^i_c\xi^c}{q^0}\bigg)
+p^0\bigg(\frac{p^i}{p^0}-\frac{p_ae^a_b\xi^be^i_c\xi^c}{p^0}\bigg)\bigg),\label{ed4}\\
&\partial_{p_i}\bigg[\frac{n_ae^a_b\omega^bn_j}{n_cn^c}\bigg]
=\frac{e^i_a{\omega}^an_j}{n_bn^b}
+\frac{n_ae^a_b{\omega}^b\delta^i_j}{n_cn^c}
-2\frac{n_ae^a_b{\omega}^bn_jg^{ic}n_c}{(n_dn^d)^2}.\label{ed5}
\end{align}
\end{lemma}
\begin{proof}
The results of this lemma are direct calculations. For instance, \eqref{ed1}
is obtained as follows:
\begin{align*}
\partial_{p_i}p^0=\partial_{p_i}\sqrt{1+p_ap^a}
=\frac{\delta^i_ap^a}{\sqrt{1+p_bp^b}}=\frac{p^i}{p^0},
\end{align*}
where $\delta^i_a$ is the Kronecker delta. The other identities are similarly obtained,
and we skip the proof.
\end{proof}

\begin{lemma}\label{lem_weight}
Let $p_*$ and $q_*$ be given. Suppose that $p_*'$ and $q_*'$ are post-collision momenta with a parameter $\omega$ which satisfies the assumption on the scattering kernel. Then, we have
\begin{eqnarray*}
\bar{g}^{ab}(p_ap_b+q_aq_b-p'_ap'_b-q'_aq'_b)\leq C,
\end{eqnarray*}
where $\bar{g}^{ab}$ is defined by \eqref{bar}.
\end{lemma}
\begin{proof}
We first consider $g^{ab}$ to take an orthonormal frame such that
\[
g^{ab}(p_ap_b+q_aq_b-p'_ap'_b-q'_aq'_b)
=|\hat{p}|^2+|\hat{q}|^2-|\hat{p}'|^2-|\hat{q}'|^2.
\]
Then, we use the energy conservation of particles to have
\[
|\hat{p}|^2+|\hat{q}|^2-|\hat{p}'|^2-|\hat{q}'|^2
=(\hat{p}^0)^2+(\hat{q}^0)^2-(\hat{p}'^0)^2-(\hat{q}'^0)^2
=2\hat{p}'^0\hat{q}'^0-2\hat{p}^0\hat{q}^0.
\]
Since $h$ is an invariant quantity in the collision process, we have
\[
h^2=(\hat{p}_\alpha-\hat{q}_\alpha)(\hat{p}^\alpha-\hat{q}^\alpha)=(\hat{p}'_\alpha-\hat{q}'_\alpha)(\hat{p}'^\alpha-\hat{q}'^\alpha),
\]
and then 
\[
h^2+(\hat{p}^0+\hat{q}^0)^2=|\hat{u}|^2+4\hat{p}^0\hat{q}^0=|\hat{u}'|^2+4\hat{p}'^0\hat{q}'^0.
\]
By the assumption on the scattering kernel, we have
\[
4\hat{p}'^0\hat{q}'^0-4\hat{p}^0\hat{q}^0=|\hat{u}|^2-|\hat{u}'|^2\leq Ct^{-\frac43}.
\]
Since $g^{ab}=t^{-\frac43}\bar{g}^{ab}$, this concludes the lemma.
\end{proof}

\section{Existence of solutions of the Einstein equations}\label{Sec_Einstein}
In this section we study the Einstein equations for prescribed matter terms.
We will assume that a distribution function $f$ is given, and the energy-momentum
tensor is derived from \eqref{rho} and \eqref{S_ab},
and will show that classical solutions to the Einstein equations \eqref{evolution1}--\eqref{evolution2}
exist globally in time and have certain asymptotic behaviours.
In this section we assume that the distribution function satisfies in an orthonormal frame
the following estimate:
\begin{equation}
f(t,p_*)=\hat{f}(t,\hat{p})\leq C_f\exp(-t^{\frac54}|\hat{p}|^2),\label{assump_distribution}
\end{equation}
where $C_f$ is a positive constant which will be determined later
Below, and in the rest of the paper, $C$ will denote a universal constant which depends only on initial
data and may change from line to line, but the constants with subscripts, such as $C_f$, will be
understood as fixed ones.

\subsection{Global existence of solutions}
Global-in-time existence of solutions to the Einstein equations with a given distribution function is easily proved by
following the arguments of \cite{CCH}, where the Einstein-Vlasov system has been studied for several different
Bianchi type symmetries. In the Vlasov case, it is usual to assume that the distribution function $f$ has a compact
support, i.e., there exists a positive number $P$ such that if $|p_a|\geq P$, then $f(t,p_*)=0$.
In this case, we can estimate the matter terms \eqref{rho} and \eqref{S_ab} as follows:
\[
\rho\leq AP^4,\quad |S_{ab}|\leq AP^5,
\]
where $A$ is a positive constant depending on $(\det g)^{-1}$ and $g_{ab}$. Hence, as long as
$(\det g)^{-1}$ and $g_{ab}$ are bounded, the matter terms are also bounded.
Iteration scheme then works well, and local-in-time existence follows.
For detailed arguments we refer to \cite{CCH}.
In the Boltzmann case, we do not assume that the distribution function has a compact support,
but instead assume that $f$ decays at infinity. If we assume \eqref{assump_distribution},
then the matter terms are estimated as follows:
\[
\rho=\int_{\bbr^3}\hat{f}(t,\hat{p})(1+|\hat{p}|^2)^{\frac12}d\hat{p}
\leq C_f\int_{\bbr^3}\exp(-t^{\frac54}|\hat{p}|^2)(1+|\hat{p}|^2)^{\frac12}d\hat{p}
\leq CC_ft^{-\frac{15}{8}},
\]
where $\rho$ is written in an orthonormal frame. On the other hand, we use
$p_a=g_{ab}e^b_c\hat{p}^c=\eta_{cb}\theta^b_a\hat{p}^c$ to estimate
$S_{ab}$ as follows:
\begin{align*}
S_{ab}&=(\det g)^{-\frac{1}{2}} \int_{\bbr^3} f(t,p_*)p_a p_b(1+g^{cd}p_{c}p_d)^{-\frac{1}{2}}dp_*
=\int_{\bbr^3}\hat{f}(t,\hat{p})p_ap_b(1+|\hat{p}|^2)^{-\frac{1}{2}}d\hat{p}\\
&=\theta^c_a\theta^d_b\int_{\bbr^3}\hat{f}(t,\hat{p})\hat{p}_c\hat{p}_d(1+|\hat{p}|^2)^{-\frac{1}{2}}d\hat{p}
=\theta^c_a\theta^d_b\hat{S}_{cd},
\end{align*}
where $\hat{p}_c=\eta_{ac}\hat{p}^a$,
and $\hat{S}_{cd}$ is easily estimated:
\begin{equation}
|\hat{S}_{cd}|\leq\int_{\bbr^3}\hat{f}(t,\hat{p})|\hat{p}_c||\hat{p}_d|(1+|\hat{p}|^2)^{-\frac{1}{2}}d\hat{p}
\leq C_f\int_{\bbr^3}\exp(-t^{\alpha}|\hat{p}|^2)
|\hat{p}|^2d\hat{p}\leq CC_ft^{-\frac{25}{8}}.\label{estS_ab}
\end{equation}
By the estimates \eqref{est_principal_minors},
the quantities $\theta^a_b$ are also bounded as long as $(\det g)^{-1}$ and $g_{ab}$ are bounded,
and we conclude that the matter terms are bounded. Applying the same iteration as in \cite{CCH},
we obtain local-in-time existence of solutions $g_{ab}$ and $k_{ab}$ to the Einstein equations.
Now, we can extend the local existence to global existence by the arguments exactly as in \cite{CCH},
and obtain the following global-in-time existence:

\begin{lemma}\label{Lem_existence_einstein}
Suppose that $\hat{f}=\hat{f}(t,\hat{p})$ is given such that it is $C^1$ and satisfies \eqref{assump_distribution}.
If $g_{ab}(t_0)$ and $k_{ab}(t_0)$ are initial data of the Einstein equations \eqref{evolution1}--\eqref{evolution2}
satisfying the constraints \eqref{constraint1}--\eqref{constraint2},
then there exist unique classical solutions $g_{ab}$ and $k_{ab}$ on $[t_0,\infty)$.
\end{lemma}

In the above lemma the constant $C_f$ is not assumed to be small.
In the next section we will show that certain asymptotic behaviours are obtained
by assuming smallness.

\subsection{Bootstrap argument}
The argument which will lead us to our main conclusions is of a bootstrap type.
One has a solution of an evolution equation and assumes that an appropriate norm of
the solution decays in a certain rate. In our case the relevant function is $F$.
Let $[t_0,T)$ with $T<\infty$ be the maximal interval on which the solution has the prescribed time decay rate.
One improves the decay rate such that the assumption would lead to a contradiction,
and concludes that the solution has the desired decay rate globally in time.

\subsubsection{Bootstrap assumption}
The bootstrap assumption is motivated by the fact that due to \cite{Lee2,Lee1,EN} we expect the solutions
to behave as solutions to the Einstein-Vlasov system.
We will assume that the shear $F$ satisfies on an interval $[t_0,T)$
\begin{align}
F\leq C_F (1+t)^{-\frac{3}{2}},\label{bootF}
\end{align}
where $C_F$ is a positive constant which will be determined later.

\subsubsection{Estimate for the Hubble variable $H$}\label{secestH}
Since we are assuming \eqref{assump_distribution}, we have from \eqref{estS_ab} that
\begin{equation}
S\leq CC_ft^{-\frac{25}{8}},\label{estS}
\end{equation}
where $S=g^{ab}S_{ab}=\eta^{ab}\hat{S}_{ab}$.
Then, we use the estimate \eqref{estH} for $H$ to get
\begin{eqnarray}\label{matter}
\frac{4\pi S}{3H^2}\leq CC_f t^{-\frac98}.
\end{eqnarray}
Integrating (\ref{H2}), we have
\begin{eqnarray*}
H=\frac{1}{\frac32 t + I}=\frac23 t^{-1} \frac{1}{1+\frac23 t^{-1} I},
\end{eqnarray*}
where
\begin{eqnarray*}
I =\int_{t_0}^t \bigg(F+\frac{4\pi S}{3H^2}\bigg)(s)ds.
\end{eqnarray*}
On the other hand, using the bootstrap assumption on $F$ and the estimate (\ref{matter}), we obtain
\begin{eqnarray*}
I \leq C(C_F+C_f),
\end{eqnarray*}
where $C$ is a positive constant depending only on the time origin $t_0$.
Thus, we have
\begin{eqnarray}
H \geq \frac23 t^{-1} \frac{1}{1+C(C_F+C_f)t^{-1}},\label{estHImproved}
\end{eqnarray}
which implies that we obtain a better estimate for $H$, namely:
\begin{eqnarray*}
H=\frac23t^{-1}\Big(1+O\Big((C_F+C_f) t^{-1}\Big)\Big).
\end{eqnarray*}

\subsubsection{Estimate of the shear $F$}
Our goal is to obtain an improved decay rate
\begin{equation}
F(t)\leq C_F(1+t)^{-2+\delta}\label{bootImproved}
\end{equation}
on the interval $[t_0,T)$ where $F$ satisfies \eqref{bootF}.
More precisely, we will show that if $F$ satisfies \eqref{bootF},
then there exists a small number $0<\delta<\frac15$ satisfying \eqref{bootImproved},
where the small number $\delta$ does not depend on $T$. This implies that $F$ satisfies \eqref{bootImproved}
globally in time, and completes the bootstrap argument.

We now show that $F$ satisfies \eqref{bootImproved}.
If this is the case, there is nothing more to do. Let us suppose the opposite, namely that
for any $0<\delta<\frac15$, there exist $t_1$ and $t_2$ such that
$t_0\leq t_1<t_2\leq T$ and
\begin{align}
F(t)&\leq C_F(1+t)^{-2+\delta},\quad t\in [t_0,t_1],\nonumber\\
F(t)&\geq C_F(1+t)^{-2+\delta},\quad t\in [t_1,t_2).\label{op}
\end{align}
We consider the second time interval $[t_1,t_2)$.
The evolution equation \eqref{F} is written as
\[
\dot{F}=-3H\bigg(1-\frac{2}{3}F-\frac{8\pi S}{9H^2}-\frac{4\pi}{3H^3}\frac{S_{ab}\sigma^{ab}}{F}\bigg)F,
\]
and the quantities on the left side are estimated as follows:
\begin{align*}
&\frac23 F\leq \frac23 C_F (1+t)^{-\frac32} \leq CC_F,\\
&\frac{8\pi S}{9H^2}\leq CC_f t^{-\frac98}  \leq CC_f,\\
&\frac{4\pi}{3H^3}\frac{S_{ab}\sigma^{ab}}{F}
\leq\frac{4\pi}{3}\frac{\sqrt{S_{ab}S^{ab}}}{H^2} \frac{\sqrt{\sigma_{ab}\sigma^{ab}}}{HF}
\leq C\frac{S}{H^2}F^{-\frac12}\leq CC_ft^{-\frac98}C_F^{-\frac{1}{2}}(1+t)^{1-\frac{\delta}{2}}
\leq CC_fC_F^{-\frac12},
\end{align*}
where in the last inequality we have used (\ref{op}) and the Cauchy-Schwarz inequality.
We now assume that the constants $C_F$ and $C_f$ are small.
Note that in the last inequality although $C_F^{-\frac12}$ might be a big quantity,
we have $C_f$ as a factor as well which is independent and can be chosen such that
the last quantity $C_fC_F^{-\frac12}$ is small enough.
As a consequence we can find a small positive number $\delta_1$ such that the following
differential inequality holds:
\begin{align*}
\dot{F}&\leq-2t^{-1}\frac{1}{1+C(C_F+C_f)}\bigg(1-C\bigg(C_F+C_f+C_fC_F^{-\frac12}\bigg)\bigg)F\\
&\leq (-2+\delta_1)t^{-1}F.
\end{align*}
Integrating the above differential inequality on $[t_1,t]$
with $t\in[t_1,t_2)$, we obtain
\[
F(t)\leq C_F(1+t_1)^{-2+\delta}\bigg(\frac{t}{t_1}\bigg)^{-2+\delta_1},
\]
which holds on $[t_1,t_2)$. Here, we have used $F(t_1)\leq C_F(1+t_1)^{-2+\delta}$.
Since $\delta$ has been chosen arbitrarily and $\delta_1$ does not depend on
$\delta$, $t_1$, and $t_2$, we may assume that $\delta$ has been chosen as $\delta=\delta_1$.
Hence, we obtain
\[
F(t)\leq C_F(1+t)^{-2+\delta_1},
\]
which holds on $[t_1,t_2)$, and this contradicts to \eqref{op}.
Consequently, we have shown that there exists a small number $0<\delta<\frac{1}{5}$ such that
the shear $F$ satisfies \eqref{bootImproved} globally in time.

\subsubsection{Optimal estimate of $F$}
We want to improve the estimate of $F$. For this reason we need an inequality in the other
direction. From (\ref{F}) we have
\begin{eqnarray*}
\dot{F} \geq -3HF+\frac{4\pi}{H^2}S_{ab}\sigma^{ab}.
\end{eqnarray*}
Implementing now the estimates of $F$ and $f$ coming from the results of the bootstrap argument,
we estimate the last term as
\begin{align*}
&\frac{4\pi}{H^2}S_{ab}\sigma^{ab}\leq\frac{8\pi}{H}SF^{\frac12}\\
&\leq 8\pi\bigg(\frac23 t^{-1} \frac{1}{1+C(C_F+C_f)}\bigg)^{-1}
CC_ft^{-\frac{25}{8}}C_F^{\frac12}(1+t)^{-1+\frac{\delta}{2}}\\
&\leq CC_fC_F^{\frac12}(1+C_F+C_f)t^{-\frac{17}{8}}(1+t)^{-1+\frac{\delta}{2}}\\
&\leq CC_fC_F^{\frac12}t^{-\frac{25}{8}+\frac{\delta}{2}},
\end{align*}
where we used \eqref{estHImproved}, \eqref{estS}, \eqref{bootImproved}, and the fact that
the constants $C_F$ and $C_f$ are small. We obtain
\begin{eqnarray*}
\dot{F} \geq -3H F - CC_fC_F^{\frac12}t^{-\frac{25}{8}+\frac{\delta}{2}},
\end{eqnarray*}
and together with the upper bound $H\leq\frac23t^{-1}$,
\begin{eqnarray*}
\dot{F} \geq -2t^{-1} F - CC_fC_F^{\frac12}t^{-\frac{25}{8}+\frac{\delta}{2}},
\end{eqnarray*}
from which we obtain
\begin{eqnarray*}
t^2 F(t)\geq t_0^2 F(t_0)-CC_fC_F^{\frac12}\int^t_{t_0}s^{-\frac{9}{8}+\frac{\delta}{2}}ds,
\end{eqnarray*}
and since $\delta<\frac15$,
\begin{eqnarray*}
F(t)\geq \Big(t_0^2F(t_0)- CC_fC_F^{\frac12}\Big)t^{-2}.
\end{eqnarray*}
Here, $F(t_0)$ should not be zero, and we take $C_f$ small enough such that
\begin{align}\label{U}
F(t)\geq\frac{1}{2}t_0^2F(t_0)t^{-2}.
\end{align}
With this lower bound of $F$, we come back to the evolution equation \eqref{F}
and estimate it again as follows:
\begin{align*}
&\frac23 F\leq CC_F (1+t)^{-2+\delta} \leq CC_Ft^{-2+\delta},\\
&\frac{8\pi S}{9H^2}\leq CC_f t^{-\frac98},\\
&\frac{4\pi}{3H^3}\frac{S_{ab}\sigma^{ab}}{F}
\leq C\frac{S}{H^2}F^{-\frac12}
\leq CC_fF(t_0)^{-\frac12}t^{-\frac18},
\end{align*}
where we used \eqref{bootImproved}, \eqref{matter}, and \eqref{U}.
We now obtain
\begin{align*}
\dot{F}&\leq -2t^{-1}\frac{1}{1+C(C_F+C_f)t^{-1}}
\bigg(1-C\bigg(C_F+C_f+C_fF(t_0)^{-\frac12}\bigg)t^{-\frac18}\bigg)F\\
&\leq-2t^{-1}\Big(1-C\bigg(C_F+C_f+C_fF(t_0)^{-\frac12}\bigg)t^{-\frac18}\Big)F,
\end{align*}
where we used the first inequality of \eqref{estHImproved},
and consequently obtain an improved estimate:
\[
F(t)\leq CF(t_0)t^{-2}.
\]
Let us summarize the results obtained in the bootstrap argument part:
\begin{lemma}\label{Lem_bootstrap}
Suppose that $F(t_0)\neq 0$ is sufficiently small.
There exists a small positive constant $C_f$ such that if $f$ is $C^1$ and satisfies \eqref{assump_distribution},
then there exists a small positive constant $\varepsilon$ such that the following estimates hold:
\begin{eqnarray*}
H(t)&=&\frac{2}{3}t^{-1}(1+O(\varepsilon t^{-1})),\\
F(t)&=&O(\varepsilon t^{-2}),
\end{eqnarray*}
where the constant $\varepsilon$ depends on $F(t_0)$ and $C_f$.
\end{lemma}
\begin{proof}
If the initial data $F(t_0)$ is small, then by continuity there must be an interval on which $F$ satisfies
$F(t)\leq C_F(1+t)^{-\frac32}$ for some $C_F$. Let $[t_0,T)$ be the maximal interval, where $T$ may be infinite.
Suppose that $T<\infty$. Since $F(t_0)$ is small, the constant $C_F$ can be chosen small,
and we have shown in this section that in fact $F(t)\leq C_F(1+t)^{-2+\delta}<C_F(1+t)^{-\frac32}$ on $[t_0,T)$.
Now, by continuity $F(t)$ remains less than $C_F(1+t)^{-\frac32}$ for a short time after $T$,
but this is a contradiction. Thus, $T=\infty$, and the estimates obtained hold globally in time.
The estimate of $H$ and the improvement of $F$ follow using these estimates.
\end{proof}

\subsection{Estimate of the metric}
In this part we estimate the metric $g_{ab}$ to obtain its asymptotic behaviour.
In the previous lemma we only have asymptotic behaviours of the Hubble variable
and the shear, but this is not enough to study the Boltzmann equation in the next section.
To obtain global-in-time existence of solutions to the Boltzmann equation, we need
an asymptotic behaviour of the metric at late times. We first estimate the metric $g_{ab}$
and $g^{ab}$, and then show that the scaled version $\bar{g}^{ab}$
is almost constant in the sense of quadratic forms.

\subsubsection{Estimate of the metric}
From the definitions (\ref{bar}) and (\ref{tracefree}) we have that
\begin{eqnarray*}
\dot{\bar{g}}_{ab}=2\Big(H-\frac23 t^{-1}\Big)\bar{g}_{ab}-2t^{-\frac{4}{3}}\sigma_{ab}.
\end{eqnarray*}
Using the usual matrix norm we obtain that
\begin{eqnarray*}
\Vert \bar{g}_{ab}(t) \Vert \leq \Vert \bar{g}_{ab}(t_0) \Vert + C\int^t_{t_0}
\Big(\Big\vert H(s)-\frac{2}{3} s^{-1}\Big\vert +(\sigma_{cd}\sigma^{cd}(s))^{\frac{1}{2}}\Big) \Vert \bar{g}_{ab}(s) \Vert ds,
\end{eqnarray*}
and with Gronwall's inequality we obtain
\begin{eqnarray*}
\Vert \bar{g}_{ab}(t) \Vert \leq \Vert \bar{g}_{ab}(t_0) \Vert \exp\bigg(
\int^t_{t_0} \Big\vert H(s)-\frac{2}{3} s^{-1}\Big\vert +(\sigma_{cd}\sigma^{cd}(s))^{\frac{1}{2}}ds\bigg)
\leq C\|\bar{g}_{ab}(t_0)\|,
\end{eqnarray*}
where we used Lemma \ref{Lem_bootstrap} as follows:
\begin{align*}
&\Big|H(s)-\frac23s^{-1}\Big|\leq C\varepsilon s^{-2},\\
&(\sigma_{ab}\sigma^{ab}(s))^{\frac12}=2HF^{\frac12}\leq C\sqrt{\varepsilon}s^{-2}.
\end{align*}
Therefore, $\bar{g}_{ab}$ is bounded for all $t\geq t_0$, and we have
\begin{eqnarray*}
\vert t^{-\frac{4}{3}}g_{ab} \vert \leq C,
\end{eqnarray*}
where we can choose the constant $C$ as a constant which is independent of $\varepsilon$.
Note that $\|\sigma_{ab}\|\leq(\sigma_{cd}\sigma^{cd})^{\frac12}\|g_{ab}\|$ by \cite{CCH}.
Then, we can conclude that
\begin{eqnarray*}
\Vert \sigma_{ab} \Vert \leq C\sqrt{\varepsilon} t^{-\frac{2}{3}},\quad
\sigma_{ab} = O(\sqrt{\varepsilon}t^{-\frac{2}{3}}).
\end{eqnarray*}
Looking again at the derivative of $\bar{g}_{ab}$ and putting the facts which have been obtained together,
we see that
\begin{eqnarray*}
\dot{\bar{g}}_{ab}=O(\sqrt{\varepsilon}t^{-2}).
\end{eqnarray*}
This is enough to conclude that there exists a constant matrix $\mathcal{G}_{ab}$ defined by
\[
\mathcal{G}_{ab}=\bar{g}_{ab}(t_0)+\int_{t_0}^\infty
2\Big(H(s)-\frac23 s^{-1}\Big)\bar{g}_{ab}(s)-2s^{-\frac{4}{3}}\sigma_{ab}(s)ds.
\]
Note that
\begin{align*}
|\mathcal{G}_{ab}-\bar{g}_{ab}(t)|
&\leq C\sqrt{\varepsilon}\int_t^\infty s^{-2}ds\leq C\sqrt{\varepsilon}t^{-1},
\end{align*}
and in this sense we obtain
\[
g_{ab}=t^{+\frac43}\Big(\mathcal{G}_{ab}+O(\sqrt{\varepsilon}t^{-1})\Big).
\]
We take the opportunity to note that in \cite{EN} there is a mistake, where instead of $O(t^{-1})$ there is a wrong $O(t^{-2})$ in (61) and (62).
By similar computations we can see that there exists a constant matrix $\mathcal{G}^{ab}$ such that
the inverse $g^{ab}$ is estimated as follows:
\[
g^{ab}=t^{-\frac43}\Big(\mathcal{G}^{ab}+O(\sqrt{\varepsilon}t^{-1})\Big).
\]
We may choose a small number $\varepsilon$ again, if necessary, and combine the above results
with Lemma \ref{Lem_bootstrap} to obtain the following proposition:
\begin{prop}\label{Prop_Einstein}
Suppose that $F(t_0)\neq 0$ is sufficiently small. There exists a small positive constant $C_f$ such that
if $f$ is $C^1$ and satisfies \eqref{assump_distribution}, then there exist a small positive constant
$\varepsilon$ and constant matrices $\mathcal{G}_{ab}$ and $\mathcal{G}^{ab}$
such that the following estimates hold:
\begin{eqnarray*}
H(t)&=&\frac23 t^{-1}(1+O(\varepsilon t^{-1})),\\
F(t)&=&O(\varepsilon t^{-2}),\\
g_{ab}(t)&=& t^{+\frac43}\Big(\mathcal{G}_{ab}+O(\varepsilon t^{-1})\Big),\\
g^{ab}(t)&=& t^{-\frac43}\Big(\mathcal{G}^{ab}+O(\varepsilon t^{-1})\Big),
\end{eqnarray*}
where the constant $\varepsilon$ depends on $F(t_0)$ and $C_f$.
\end{prop}

\subsubsection{Estimate of the scaled metric}
In this part we estimate the scaled metric $\bar{g}^{ab}$ defined by $g^{ab}=t^{-\frac43}\bar{g}^{ab}$.
From algebra we know that in the sense of \emph{quadratic forms}
$\sigma^{ab} \leq (\sigma_{cd} \sigma^{cd})^{\frac12} g^{ab}$, i.e., for any $p_*$ the following holds:
\begin{eqnarray}\label{qf}
\sigma^{ab}p_ap_b \leq (\sigma_{cd} \sigma^{cd})^{\frac12} g^{ab}p_ap_b.
\end{eqnarray}
Introduce now for technical reasons $\varepsilon_1$ and consider the following expression where
we use the fact that $\dot{g}^{ab}=2k^{ab}=2(\sigma^{ab}-H g^{ab})$:
\begin{eqnarray*}
\frac{d}{dt}(t^{-\varepsilon_1}\bar{g}^{ab})
=\Big(-\varepsilon_1t^{-1}
+\frac43t^{-1}-2H\Big)t^{-\varepsilon_1}\bar{g}^{ab}
+2t^{-\varepsilon_1+\frac43}\sigma^{ab}.
\end{eqnarray*}
Using expression (\ref{qf}) we obtain
\begin{align*}
\frac{d}{dt}(t^{-\varepsilon_1}\bar{g}^{ab}p_ap_b)
&=\Big(-\varepsilon_1t^{-1}+\frac43t^{-1}-2H \Big)t^{-\varepsilon_1}\bar{g}^{ab}p_ap_b
+2t^{-\varepsilon_1+\frac43}\sigma^{ab}p_ap_b\\
&\leq(-\varepsilon_1t^{-1}+C\varepsilon t^{-2})t^{-\varepsilon_1}\bar{g}^{ab}p_ap_b
+4t^{-\varepsilon_1}HF^{\frac12}\bar{g}^{ab}p_ap_b.
\end{align*}
Using the results of Proposition \ref{Prop_Einstein}, we have
\begin{align*}
\frac{d}{dt}(t^{-\varepsilon_1}\bar{g}^{ab}p_ap_b)
&\leq\Big(C\sqrt{\varepsilon} t^{-2} -\varepsilon_1t^{-1}\Big)(t^{-\varepsilon_1}\bar{g}^{ab}p_ap_b).
\end{align*}
We can choose a positive number $\varepsilon_1$ appropriately such that the following inequality holds:
\begin{align}\label{est_gbar}
\frac{d}{dt}(t^{-\varepsilon_1}\bar{g}^{ab}p_ap_b) \leq 0,
\end{align}
where the constant $\varepsilon_1$ depends on the constant $\varepsilon$ of Proposition \ref{Prop_Einstein},
which in turn depends on $F(t_0)$ and $C_f$.
Note that $\varepsilon_1$ can be chosen small, since the constant $\varepsilon$ is small
and $t^{-2}< t^{-1}$ at late times.
This estimate can also be stated as $d(t^{-\varepsilon_1}\bar{g}^{ab})/dt\leq 0$ in the sense of quadratic forms,
and will be used to estimate the distribution function of the Boltzmann equation.

\subsubsection{Estimate of the determinant}
We apply \eqref{det} to the results of Proposition \ref{Prop_Einstein}.
Integrating \eqref{det}, we get
\[
\det g(t)=\det g(t_0)\exp\Big(6\int_{t_0}^tH(s)ds\Big),
\]
and applying it to Proposition \ref{Prop_Einstein},
\[
\det g(t_0) e^{-C\varepsilon}\bigg(\frac{t}{t_0}\bigg)^4
\leq \det g(t)
\leq \det g(t_0) e^{C\varepsilon}\bigg(\frac{t}{t_0}\bigg)^4.
\]
Since $\varepsilon$ is a small number, we may simply write
\begin{equation}\label{est_det}
\frac{1}{2}\bigg(\frac{t}{t_0}\bigg)^4
\leq \frac{\det g(t)}{\det g(t_0)}
\leq 2\bigg(\frac{t}{t_0}\bigg)^4.
\end{equation}
Applying this to the estimates \eqref{est_principal_minors}, we obtain
\[
\frac{1}{g_{11}}\leq Ct^{-\frac43},\quad
\frac{1}{g_{11}g_{22}-(g_{12})^2}\leq Ct^{-\frac83}.
\]
From the explicit expressions of $e^a_b$ and $\theta^a_b$ in \eqref{e^a_b} and \eqref{theta^a_b},
we derive
\begin{equation}\label{est_e^a_b}
|e^a_b|\leq Ct^{-\frac23},\quad |\theta^a_b|\leq Ct^{\frac23},
\end{equation}
where the constants $C$ depend only on initial data.

\section{Existence of solutions of the Boltzmann equation}\label{Sec_Boltzmann}
In this section we study the Boltzmann equation. We will assume that a spacetime is given and will show that classical solutions to the Boltzmann equation exist globally in time. The spacetime will be assumed to satisfy the properties of Proposition \ref{Prop_Einstein} at late times. To be precise, we assume that the metric satisfies the following estimates.\bigskip

\noindent{\sf (A)} {\bf Assumption on the spatial metric.}
There exists a positive constant $C$ such that the metric $g_{ab}$ and its inverse $g^{ab}$ satisfy $|\bar{g}_{ab}|\leq C$ and $|\bar{g}^{ab}|\leq C$ for $g_{ab}=t^{\frac43}\bar{g}_{ab}$ and $g^{ab}=t^{-\frac43}\bar{g}^{ab}$. In an orthonormal frame, the metric can be written as $g_{ab}e^a_ce^b_d=\eta_{cd}$ or $g_{ab}=\theta^c_a\theta^d_b\eta_{cd}$, and they satisfy $|e^a_b|\leq Ct^{-\frac23}$ and $|\theta^a_b|\leq Ct^{\frac23}$. Moreover, there exists a small number $0<\varepsilon_1<\frac{1}{12}$ such that $d(t^{-\varepsilon_1}\bar{g}^{ab})/dt\leq 0$ in the sense of quadratic forms.\bigskip

Under the assumption {\sf (A)} on the metric, we will prove global-in-time existence of solutions to the Boltzmann equation. The following inequalities are direct consequences of Lemma \ref{Lem_partial_p}.

\begin{lemma}
Suppose that the metric satisfies {\sf (A)}. Then,
the following estimates hold:
\begin{eqnarray}
&&|\partial_{p_i}p^0|\leq Ct^{-\frac23},\quad
|\partial_{p_i}h|\leq\frac{Cq^0t^{-\frac23}}{h},\quad
|\partial_{p_i}\sqrt{s}|\leq\frac{Cq^0t^{-\frac23}}{\sqrt{s}},\label{edR1}\\
&&|\partial_{p_i}h|\leq C(p^0)^{\frac12}(q^0)^{\frac32}t^{-\frac23},\quad
|\partial_{p_i}\sqrt{s}|\leq C(p^0)^{\frac12}(q^0)^{\frac32}t^{-\frac23},\label{edR2}\\
&&\bigg|\partial_{p_i}\sqrt{(n^0)^2-(n_ae^a_b\xi^b)^2}\bigg|
\leq\frac{Cn^0t^{-\frac23}}{\sqrt{(n^0)^2-(n_ae^a_b\xi^b)^2}},\label{edR3}\\
&&\bigg|\partial_{p_i}\bigg[\frac{n_ae^a_b\omega^bn_j}{n_cn^c}\bigg]\bigg|
\leq\frac{C}{\sqrt{n_an^a}},\label{edR4}
\end{eqnarray}
where the constants $C$ do not depend on the metric.
\end{lemma}
\begin{proof}
The quantity in \eqref{ed1} can be written as
\[
\partial_{p_i}p^0=\frac{p^i}{p^0}=\frac{e^i_a\hat{p}^a}{\sqrt{1+|\hat{p}|^2}},
\]
and then the assumption {\sf (A)} gives the first inequality.
The other inequalities of \eqref{edR1} are similarly proved.
For the first inequality of \eqref{edR2}, we have
\begin{align*}
\bigg|\frac{p^i}{p^0}-\frac{q^i}{q^0}\bigg|
&=\bigg|\frac{e^i_a\hat{p}^a}{\sqrt{1+|\hat{p}|^2}}
-\frac{e^i_a\hat{q}^a}{\sqrt{1+|\hat{q}|^2}}\bigg|
\leq Ct^{-\frac23}|\hat{p}-\hat{q}|\\
&=Ct^{-\frac23}\sqrt{(p_a-q_a)(p^a-q^a)}
\leq Ct^{-\frac23}h\sqrt{p^0q^0},
\end{align*}
where we used \eqref{e2} in the last inequality, and then from \eqref{ed2} we obtain
\[
|\partial_{p_i}h|=\frac{q^0}{h}\bigg|\frac{p^i}{p^0}-\frac{q^i}{q^0}\bigg|
\leq C(p^0)^{\frac12}(q^0)^{\frac32}t^{-\frac23}.
\]
The second estimate of \eqref{edR2} is similarly obtained, since $h\leq \sqrt{s}$,
and the inequalities \eqref{edR3} and \eqref{edR4} are also similarly proved from Lemma \ref{Lem_partial_p}.
\end{proof}

\begin{lemma}\label{lem_partial2}
Consider a momentum $p'_i$ in \eqref{p'2}. The following estimate holds:
\[
|\partial_{p_k}p'_i|\leq Cp^0(q^0)^4,
\]
where the constant $C$ does not depend on the metric.
\end{lemma}
\begin{proof}
Recall that $p'_i$ is written in the representation \eqref{p'2} as
\[
p'_i=\frac{p_i+q_i}{2}+\frac{h}{2}\frac{n^0g_{ia}e^a_b\xi^b}
{\sqrt{(n^0)^2-(n_ce^c_d\xi^d)^2}},
\]
and then $\partial_{p_k}p'_i$ is given by
\begin{align*}
\partial_{p_k}p'_i
&=\frac{\delta_i^k}{2}
+\frac{\partial_{p_k}h}{2}\frac{n^0g_{ia}e^a_b\xi^b}{\sqrt{(n^0)^2-(n_ce^c_d\xi^d)^2}}
+\frac{h}{2}\frac{(\partial_{p_k}p^0)g_{ia}e^a_b\xi^b}{\sqrt{(n^0)^2-(n_ce^c_d\xi^d)^2}}\\
&-\frac{h}{2}\frac{n^0g_{ia}e^a_b\xi^b}{(n^0)^2-(n_ce^c_d\xi^d)^2}
\partial_{p_k}\sqrt{(n^0)^2-(n_ce^c_d\xi^d)^2}.
\end{align*}
We now estimate each term separately. The first term is clearly bounded by $1/2$.
For the second term, we use \eqref{edR2} and \eqref{e3} to obtain
\begin{align*}
\bigg|\frac{(\partial_{p_k}h)n^0g_{ia}e^a_b\xi^b}{2\sqrt{(n^0)^2-(n_ce^c_d\xi^d)^2}}\bigg|
&\leq Ct^{-\frac23}\frac{(p^0)^{\frac12}(q^0)^{\frac32}(p^0+q^0)|g_{ia}e^a_b\xi^b|\sqrt{q^0}}{\sqrt{p^0}}\\
&\leq C\Big(p^0(q^0)^2+(q^0)^3\Big),
\end{align*}
where we used {\sf (A)}.
For the third term, we use \eqref{e1}, \eqref{edR1}, and \eqref{e3} to get
\begin{align*}
\bigg|\frac{h}{2}\frac{(\partial_{p_k}p^0)g_{ia}e^a_b\xi^b}{\sqrt{(n^0)^2-(n_ce^c_d\xi^d)^2}}\bigg|
&\leq Ct^{-\frac23}\frac{\sqrt{p^0q^0}|g_{ia}e^a_b\xi^b|\sqrt{q^0}}{\sqrt{p^0}}
\leq Cq^0.
\end{align*}
For the fourth term, we use \eqref{e1}, \eqref{e3}, and \eqref{edR3} to get
\begin{align*}
&\bigg|\frac{h}{2}\frac{n^0g_{ia}e^a_b\xi^b}{(n^0)^2-(n_ce^c_d\xi^d)^2}
\partial_{p_k}\sqrt{(n^0)^2-(n_ce^c_d\xi^d)^2}\bigg|\\
&\leq Ct^{-\frac23}\frac{\sqrt{p^0q^0}(p^0+q^0)|g_{ia}e^a_b\xi^b|q^0(p^0+q^0)\sqrt{q^0}}{p^0
\sqrt{p^0}}\\&\leq C(q^0)^2\bigg(p^0+q^0+\frac{(q^0)^2}{p^0}\bigg).
\end{align*}
We combine the above estimates to obtain
\[
|\partial_{p_k}p'_i|\leq Cp^0(q^0)^4,
\]
where we used the fact that $p^0\geq 1$ and $q^0\geq 1$, and this completes the proof.
\end{proof}

\begin{lemma}\label{lem_partial1}
Consider a momentum $p'_i$ in \eqref{p'1}. The following estimate holds:
\[
|\partial_{p_k}p'_i|\leq C\bigg(1+\frac{p^0}{|\hat{u}|}
+\frac{p^0}{|\hat{n}|}+\frac{(p^0)^2}{|\hat{u}|^2}
\bigg)(q^0)^3,
\]
where the constant $C$ does not depend on the metric.
\end{lemma}
\begin{proof}
Recall that $p'_i$ is written in \eqref{p'1} as
\[
p'_i=
\frac{p_i+q_i}{2}+\frac{h}{2}\bigg(
g_{ia}e^a_b\omega^b
+\bigg(\frac{n^0}{\sqrt{s}}-1\bigg)
\frac{n_ae^a_b\omega^bn_i}{g^{cd}n_cn_d}
\bigg),
\]
and directly differentiate this with respect to $p_k$ to obtain
\begin{align*}
\partial_{p_k}p'_i&=\frac{\delta^k_i}{2}
+\frac{(\partial_{p_k}h)}{2}\bigg(
g_{ia}e^a_b\omega^b
+\bigg(\frac{n^0}{\sqrt{s}}-1\bigg)
\frac{n_ae^a_b\omega^bn_i}{g^{cd}n_cn_d}
\bigg)\\
&+\frac{h}{2}\bigg(
\partial_{p_k}\bigg[\frac{n^0}{\sqrt{s}}-1\bigg]
\frac{n_ae^a_b\omega^bn_i}{g^{cd}n_cn_d}
+\bigg(\frac{n^0}{\sqrt{s}}-1\bigg)
\partial_{p_k}\bigg[\frac{n_ae^a_b\omega^bn_i}{g^{cd}n_cn_d}\bigg]
\bigg).
\end{align*}
Note that in an orthonormal frame,
\[
\bigg|\frac{n_ae^a_b\omega^bn_i}{g^{cd}n_cn_d}\bigg|
= \frac{|\hat{n}_b\omega^b \theta^c_i\hat{n}_c|}{|\hat{n}|^2}
\leq Ct^{\frac23}.
\]
We now estimate the above quantities separately.
The first term is bounded by $1/2$, and for the second and the third terms
we use \eqref{edR1} and \eqref{e2} to get
\[
\bigg|\frac{(\partial_{p_k}h)}{2}
g_{ia}e^a_b\omega^b\bigg|\leq \frac{Cq^0}{h},
\]
and
\[
\bigg|\frac{(\partial_{p_k}h)}{2}
\bigg(\frac{n^0}{\sqrt{s}}-1\bigg)
\frac{n_ae^a_b\omega^bn_i}{g^{cd}n_cn_d}\bigg|
\leq Ct^{-\frac23}\frac{q^0n^0}{h\sqrt{s}}
\bigg|\frac{n_ae^a_b\omega^bn_i}{g^{cd}n_cn_d}\bigg|\leq \frac{Cq^0n^0}{h^2}.
\]
For the fourth term, we note that
\[
\partial_{p_k}\bigg[\frac{n^0}{\sqrt{s}}-1\bigg]=
\frac{\partial_{p_k}p^0}{\sqrt{s}}-\frac{n^0}{s}\partial_{p_k}\sqrt{s},
\]
and then we can use \eqref{edR1}. Then, we have
\[
\bigg|\frac{h}{2}
\partial_{p_k}\bigg[\frac{n^0}{\sqrt{s}}-1\bigg]
\frac{n_ae^a_b\omega^bn_i}{g^{cd}n_cn_d}\bigg|
\leq Ch\bigg(\frac{1}{\sqrt{s}}
+\frac{n^0q^0}{s\sqrt{s}}\bigg)
\leq C\Big(1+\frac{n^0q^0}{s}\Big).
\]
For the fifth term, we use \eqref{e1} and \eqref{edR4} to get
\begin{align*}
\bigg|\frac{h}{2}
\bigg(\frac{n^0}{\sqrt{s}}-1\bigg)
\partial_{p_k}\bigg[\frac{n_ae^a_b\omega^bn_i}{g^{cd}n_cn_d}\bigg]
\bigg|\leq \frac{Cn^0}{\sqrt{n_an^a}}.
\end{align*}
We combine the above estimates and apply $n^0\leq 2p^0q^0$ and \eqref{e2} to get
\[
|\partial_{p_k}p'_i|\leq C\bigg(1+\frac{p^0}{|\hat{u}|}
+\frac{(p^0)^2}{|\hat{u}|^2}
+\frac{p^0}{|\hat{n}|}\bigg)(q^0)^3,
\]
and this completes the proof.
\end{proof}

Now we are in position to construct classical solutions to the Boltzmann equation.
We choose the weight function as
\[
\exp(t^{-\varepsilon_1}\bar{g}^{ab}p_ap_b)=\exp(t^{-\varepsilon_1+\frac43}|\hat{p}|^2),
\]
and define
\begin{equation*}
\|f(t)\|=\sup_{\substack{k=1,2,3\\n=0,1}}
\Big\{|\exp(\tau^{-\varepsilon_1}\bar{g}^{ab}p_ap_b)\partial_{p_k}^nf(\tau,p_*)|
:0\leq\tau\leq t,\ p_*\in\bbr^3\Big\},
\end{equation*}
where $\varepsilon_1$ is the constant given in the assumption {\sf (A)}.

\begin{lemma}\label{lem_est1}
Suppose that a spacetime is given and the metric satisfies the assumption {\sf (A)}.
If $f$ is a solution of the Boltzmann equation, then $f$ satisfies the following estimate:
\[
\exp(t^{-\varepsilon_1}\bar{g}^{ab}p_ap_b)f(t,p_*)
\leq \|f(t_0)\|+C\|f(t)\|^2.
\]
\end{lemma}
\begin{proof}
To the equation in the form of \eqref{boltzmann1} we multiply
the weight function to obtain
\begin{align*}
&\partial_t\Big[\exp(t^{-\varepsilon_1} \bar{g}^{ab} p_ap_b)f(t,p_*)\Big]\\
&=\frac{d}{dt}(t^{-\varepsilon_1}\bar{g}^{ab}p_ap_b)\exp(t^{-\varepsilon_1} \bar{g}^{cd} p_cp_d)f(t,p_*)
+\exp(t^{-\varepsilon_1} \bar{g}^{ab} p_ap_b)Q(f,f)(t,p_*)\\
&\leq (\det g)^{-\frac12} \iint v_M\sigma(p_*,q_*,\omega)f(p_*')f(q_*')
\exp(t^{-\varepsilon_1} \bar{g}^{ab} p_ap_b)
d\omega dq_*,
\end{align*}
where we ignored the loss term $Q_-$ and used the assumption {\sf (A)} such that
$d(t^{-\varepsilon_1}\bar{g}^{ab}p_ap_b)/dt\leq 0$. We now use
Lemma \ref{lem_weight} as follows:
\begin{align*}
&f(p_*')f(q_*')\exp(t^{-\varepsilon_1} \bar{g}^{ab} p_ap_b)\\
&=\Big(\exp(t^{-\varepsilon_1} \bar{g}^{ab} p'_ap'_b)f(p_*')\Big)
\Big(\exp(t^{-\varepsilon_1} \bar{g}^{ab} q'_aq'_b)f(q_*')\Big)\\
&\quad\times\exp(t^{-\varepsilon_1} \bar{g}^{ab}(p_ap_b+q_aq_b-p'_ap'_b-q'_aq'_b))
\exp(-t^{-\varepsilon_1} \bar{g}^{ab} q_aq_b)\\
&\leq C\Big(\exp(t^{-\varepsilon_1} \bar{g}^{ab} p'_ap'_b)f(p_*')\Big)
\Big(\exp(t^{-\varepsilon_1} \bar{g}^{ab} q'_aq'_b)f(q_*')\Big)\exp(-t^{-\varepsilon_1} \bar{g}^{ab} q_aq_b)\\
&\leq C\|f(t)\|^2\exp(-t^{-\varepsilon_1} \bar{g}^{ab} q_aq_b).
\end{align*}
Since $v_M$ and $\sigma(p_*,q_*,\omega)$ are bounded, we have
\begin{align*}
&(\det g)^{-\frac12}\iint v_M\sigma(p_*,q_*,\omega)f(p_*')f(q_*')
\exp(t^{-\varepsilon_1} \bar{g}^{ab} p_ap_b)
d\omega dq_*\\
&\leq C(\det g)^{-\frac12}\|f(t)\|^2\int \exp(-t^{-\varepsilon_1} \bar{g}^{ab} q_aq_b)dq_*\\
&=C\|f(t)\|^2\int \exp(-t^{-\varepsilon_1+\frac43}|\hat{q}|^2)d\hat{q}
\leq C\|f(t)\|^2t^{-2+\frac{3}{2}\varepsilon_1}.
\end{align*}
Since $0\leq\varepsilon_1<\frac{1}{12}$ is small, the quantity $t^{-2+\frac{3}{2}\varepsilon_1}$ is integrable,
and we obtain that:
\[
\exp(t^{-\varepsilon_1}\bar{g}^{ab}p_ap_b)f(t,p_*)
\leq \|f(t_0)\|+C\|f(t)\|^2,
\]
and this completes the proof.
\end{proof}

\begin{lemma}\label{Lem_estQ}
Suppose that a spacetime is given and the metric satisfies the assumption {\sf (A)}.
If $f$ is a solution of the Boltzmann equation, then $f$ satisfies the following estimate:
\[
|\exp(t^{-\varepsilon_1}\bar{g}^{ab}p_ap_b)\partial_{p_k}Q(f,f)|\leq
C \|f(t)\|^2t^{-2+\frac32\varepsilon_1}.
\]
\end{lemma}
\begin{proof}
In the proof of this lemma we need to consider the second expression of the Boltzmann equation \eqref{boltzmann2}. We first consider \eqref{boltzmann1} and take derivative to it with respect to $p_k$. Multiplying the weight function, we obtain the following quantities:
\begin{align*}
&\exp(t^{-\varepsilon_1}\bar{g}^{ab}p_ap_b)\partial_{p_k}Q(f,f)\\
&=(\det g)^{-\frac12}\iint \partial_{p_k}\Big[v_M\sigma({p_*,q_*,\omega})\Big]
\Big(f(p'_*)f(q'_*)-f(p_*)f(q_*)\Big)\exp(t^{-\varepsilon_1}\bar{g}^{ab}p_ap_b)d{\omega}dq_*\\
&+(\det g)^{-\frac12}\iint v_M\sigma(p_*,q_*,\omega)
\partial_{p_k}\Big[f(p'_*)f(q'_*)\Big]\exp(t^{-\varepsilon_1}\bar{g}^{ab}p_ap_b)d{\omega}dq_*\\
&-(\det g)^{-\frac12}\iint v_M\sigma(p_*,q_*,\omega)
\partial_{p_k}f(p_*)f(q_*)\exp(t^{-\varepsilon_1}\bar{g}^{ab}p_ap_b)d{\omega}dq_*
=:J_1+J_2+J_3.
\end{align*}
The estimate of $J_3$ is easily obtained and is
\[
|J_3|\leq C\|f(t)\|^2t^{-2+\frac32\varepsilon_1}.
\]
For $J_1$, we note that the scattering kernel is bounded and 
\[
|\partial_{p_k}v_M|\leq \frac{\sqrt{s}}{p^0q^0}|\partial_{p_k}h|
+\frac{h}{p^0q^0}|\partial_{p_k}\sqrt{s}|
+\frac{h\sqrt{s}}{(p^0)^2q^0}|\partial_{p_k}p^0|
\leq Cq^0t^{-\frac23}\leq Cq^0,
\]
where \eqref{e1} and \eqref{edR2} has been used for the first and the second quantities,
and \eqref{e1} and \eqref{edR1} for the third quantity.
Following the same calculation as in Lemma \ref{lem_est1}, we can see that $J_1$ is estimated as
\begin{align*}
|J_1|&\leq C(\det g)^{-\frac12}\|f(t)\|^2\int q^0\exp(-t^{-\varepsilon_1}\bar{g}^{ab}q_aq_b)dq_*\\
&=C\|f(t)\|^2\int \sqrt{1+|\hat{q}|^2}\exp(-t^{-\varepsilon_1+\frac43}|\hat{q}|^2)d\hat{q}\\
&=C\|f(t)\|^2t^{-2+\frac32\varepsilon_1}
\int \sqrt{1+t^{-\frac43+\varepsilon_1}|z|^2}\exp(-|z|^2)dz
\leq C\|f(t)\|^2t^{-2+\frac32\varepsilon_1},
\end{align*}
where $C$ may depend on initial data.

We now estimate $J_2$. Since the derivative term is expanded as
\[
\partial_{p_k}\Big[f(p'_*)f(q'_*)\Big]
=(\partial_{p_i}f)(p'_*)(\partial_{p_k}p'_i)f(q'_*)
+f(p'_*)(\partial_{p_i}f)(q'_*)(\partial_{p_k}q'_i),
\]
the quantity $J_2$ is separated into two integrals. The calculations are exactly same,
and we only consider the first case, i.e., $J_2=J_{21}+J_{22}$, where
\[
J_{21}=(\det g)^{-\frac12}\iint v_M\sigma(p_*,q_*,\omega)
(\partial_{p_i}f)(p'_*)(\partial_{p_k}p'_i)f(q'_*)
\exp(t^{-\varepsilon_1}\bar{g}^{ab}p_ap_b)d{\omega}dq_*.
\]
Note that the integral $J_{21}$ can be written as
\begin{align*}
J_{21}&=(\det g)^{-\frac12} \int_{\bbr^3} v_M K\exp(t^{-\varepsilon_1}\bar{g}^{ab}p_ap_b)dq_*,
\end{align*}
where
\begin{equation}\label{K1}
K=\int_{\bbs^2}\sigma(p_*,q_*,\omega)
(\partial_{p_i}f)(p'_*)(\partial_{p_k}p'_i)f(q'_*)d{\omega}.
\end{equation}
The second form of the Boltzmann equation \eqref{boltzmann2} is applied to the quantity $K$.
By the transformation ${\omega}\mapsto\xi$ with the Jacobian \eqref{appendix6}, we obtain
\begin{equation}\label{K2}
K=\int_{\bbs^2}\sigma(p_*,q_*,\xi)
(\partial_{p_i}f)(p'_*)(\partial_{p_k}p'_i)f(q'_*)\frac{\sqrt{s}(n^0)^2d\xi}
{((n^0)^2-(n_ae^a_b\xi^b)^2)^{3/2}},
\end{equation}
where $p'_*$ and $q'_*$ are now understood as \eqref{p'2}.
The estimates of $J_{21}$ are separated into three cases as follows:
(a) $|\hat{p}|^2\leq 1$,
(b) $|\hat{p}|^2\geq 1$ and $|\hat{p}|^2\leq 4|\hat{q}|^2$,
(c) $|\hat{p}|^2\geq 1$ and $|\hat{p}|^2\geq 4|\hat{q}|^2$.\bigskip

\noindent (a) In the first case, $p^0$ is bounded, and we can use \eqref{K2}
with the estimate of Lemma \ref{lem_partial2}:
\[
p^0=\sqrt{1+|\hat{p}|^2}\leq \sqrt{2},\quad
|\partial_{p_k}p'_i|\leq Cp^0(q^0)^4\leq C(q^0)^4.
\]
Together with \eqref{e1} and \eqref{e3}, the integral $K$ in \eqref{K2}
is estimated as follows:
\[
|K|\leq C\int_{\bbs^2}\sigma(p_*,q_*,\xi)|\partial_{p_i}f(p'_*)|(q^0)^6f(q'_*)d\xi.
\]
The integral $J_{21}$ is now estimated by the same argument as in the previous lemma:
\[
|J_{21}|\leq C\|f(t)\|^2t^{-2+\frac32\varepsilon_1}.
\]

\noindent (b) In the second case, we use again the second form of $K$ in \eqref{K2}. Then, we have
\[
p^0=\sqrt{1+|\hat{p}|^2}\leq \sqrt{1+4|\hat{q}|^2}\leq 2q^0,\quad
|\partial_{p_k}p'_i|\leq Cp^0(q^0)^4\leq C(q^0)^5.
\]
Together with \eqref{e1} and \eqref{e3}, we have
\[
|K|\leq C\int_{\bbs^2}\sigma(p_*,q_*,\xi)|\partial_{p_i}f(p'_*)|(q^0)^7f(q'_*)d\xi,
\]
and therefore again
\[
|J_{21}|\leq C \|f(t)\|^2t^{-2+\frac32\varepsilon_1}.
\]

\noindent (c) In the third case, we use \eqref{K1}
with the estimate of Lemma \ref{lem_partial1}.
Note that
\[
p^0=\sqrt{1+|\hat{p}|^2}\leq\sqrt{2}|\hat{p}|,\quad
|\hat{p}\pm\hat{q}|\geq|\hat{p}|-|\hat{q}|\geq\frac12|\hat{p}|.
\]
Hence, Lemma \ref{lem_partial1} implies that
\[
|\partial_{p_k}p'_i|\leq C(q^0)^3.
\]
The quantity $K$ in \eqref{K1} is estimated as
\[
|K|\leq C\int_{\bbs^2}\sigma({p_*,q_*,\omega})|\partial_{p_i}f(p'_*)|(q^0)^3f(q'_*)d\omega,
\]
and by the similar calculations we obtain
\[
|J_{21}|\leq C \|f(t)\|^2t^{-2+\frac32\varepsilon_1}.
\]
We combine the above estimates to get
\[
|\exp(t^{-\varepsilon_1}\bar{g}^{ab}p_ap_b)\partial_{p_k}Q(f,f)|\leq
C \|f(t)\|^2t^{-2+\frac32\varepsilon_1},
\]
which completes the proof.
\end{proof}

\begin{lemma}\label{lem_est2}
Suppose that a spacetime is given and the metric satisfies the assumption {\sf (A)}.
If $f$ is a solution of the Boltzmann equation, then $f$ satisfies the following estimate:
\[
|\exp(t^{-\varepsilon_1}\bar{g}^{ab}p_ap_b)\partial_{p_k}f(t,p_*)|\leq \|f(t_0)\|
+C\|f(t)\|^2.
\]
\end{lemma}
\begin{proof}
For simplicity, let us write
\[
w=t^{-\varepsilon_1}\bar{g}^{ab}p_ap_b,\quad
\phi=e^w\partial_{p_k}f,\quad
M=e^w\partial_{p_k}Q(f,f).
\]
Then, the Boltzmann equation can be rewritten as
\[
\partial_t\phi=(\partial_tw)\phi+M.
\]
Note that $w$ is nonnegative and decreasing by \eqref{est_gbar}.
Multiplying $e^{-w}$ to the above ODE and integrating it from $t_0$ to $t$, we have
\[
e^{-w(t)}\phi(t)=e^{-w(t_0)}\phi(t_0)+\int_{t_0}^te^{-w(s)}M(s)ds,
\]
and then
\[
|\phi(t)|\leq |\phi(t_0)|+\int_{t_0}^t|M(s)|ds.
\]
Applying Lemma \ref{Lem_estQ}, we conclude
\[
|\exp(t^{-\varepsilon_1}\bar{g}^{ab}p_ap_b)\partial_{p_k}f(t,p_*)|
\leq\|f(t_0)\|+C\|f(t)\|^2,
\]
which completes the proof.
\end{proof}

It is now easy to prove the global-in-time existence of classical solutions to the Boltzmann equation.
We combine Lemma \ref{lem_est1} and \ref{lem_est2} to obtain
\[
\|f(t)\|\leq \|f(t_0)\|+C\|f(t)\|^2,
\]
and by the well-known arguments as in \cite{Glassey,G06,G01,IS84,Lee2,S101},
we obtain the global-in-time existence of classical solutions for small initial data.

\begin{prop}\label{Prop_Boltzmann}
Suppose that a spacetime is given and the metric satisfies the assumption {\sf (A)}.
There exists a small positive constant $\varepsilon_0$ such that if $f(t_0)$ is $C^1$
and satisfies $\|f(t_0)\|<\varepsilon_0$,
then there exists a unique classical solution to the Boltzmann equation
corresponding to the initial data $f(t_0)$ satisfying
$\|f(t)\|\leq C\|f(t_0)\|$ on $[t_0,\infty)$. The solution $f$ is nonnegative and satisfies
in an orthonormal frame
\[
\hat{f}(t,\hat{p})\leq C\|f(t_0)\|\exp(-t^{-\varepsilon_1+\frac43}|\hat{p}|^2),
\]
where $\varepsilon_1$ is the constant given in the assumption {\sf (A)}.
\end{prop}

\section{Main results}
We are now ready to prove global-in-time existence of classical solutions to the Einstein-Boltzmann
system \eqref{evolution1}--\eqref{S_ab}. The existence is proved by a standard iteration method.
Suppose that initial data
$g_{ab}(t_0)$, $k_{ab}(t_0)$, and $f(t_0)$ are given such that $F(t_0)\neq 0$ is sufficiently small,
and define an iteration for $\{g_n\}$, $\{k_n\}$, and $\{f_n\}$ as follows.
Let $(g_0)_{ab}(t)=t^{\frac43}\bar{g}_{ab}(t_0)$ and
$(k_0)_{ab}(t)=k_{ab}(t_0)$.
Then, there exists $(\theta_0)^a_b$ satisfying $(g_0)_{ab}=(\theta_0)^c_a(\theta_0)^d_b\eta_{cd}$,
which is given by $(\theta_0)^a_b(t)=t^{\frac23}\bar{\theta}^a_b(t_0)$,
and let $(e_0)^a_b$ be the inverse of $(\theta_0)^a_b$.
Here, $\bar{\theta}^a_b(t_0)$ is uniquely
determined by $(\bar{g}_0)_{ab}(t_0)=(\bar{\theta}_0)^c_a(t_0)(\bar{\theta}_0)^d_b(t_0)\eta_{cd}$.
Note that $(\bar{g}_0)_{ab}(t)=\bar{g}_{ab}(t_0)$, i.e., $(\bar{g}_0)_{ab}$ is a constant matrix,
hence clearly we have $d(t^{-\varepsilon_1}(\bar{g}_0)^{ab})/dt\leq 0$. In other words,
$g_0$ satisfies the assumption {\sf (A)} of Section \ref{Sec_Boltzmann}.
Then, by Proposition \ref{Prop_Boltzmann}, there exists a small positive constant $\varepsilon_0$
such that if $\|f(t_0)\|<\varepsilon_0$, then there exists a unique classical solution $f_0$,
which is the solution of the Boltzmann equation in a given spacetime with metric $g_0$
and satisfies $\|f_0(t)\|\leq C\|f(t_0)\|$.
Now, suppose that $f_n$ is given such that $\|f_n(t)\|\leq C\|f(t_0)\|$ with $\|f(t_0)\|<\varepsilon_0$.
Note that, since $\varepsilon_1<\frac{1}{12}$, we have in an orthonormal frame,
$\hat{f}_n(t,\hat{p})\leq C\|f(t_0)\|\exp(-t^{\frac54}\|\hat{p}|^2)$.
We now apply Proposition \ref{Prop_Einstein}.
Since we have assumed that $\|F(t_0)\|\neq 0$ is sufficiently small, there exists a small positive
constant $C_f$ such that if $f_n$ satisfies \eqref{assump_distribution}, then $g_{n+1}$ and $k_{n+1}$ exist
globally in time, which are the solutions of ODEs, which result when $g$ and $k$ of
\eqref{evolution1}--\eqref{evolution2} are replaced by $g_{n+1}$ and $k_{n+1}$, respectively,
and $\rho$ and $S_{ab}$ are constructed by $f_n$ via \eqref{rho} and \eqref{S_ab}, respectively.
Taking $\varepsilon_0$ smaller, if necessary, we can conclude that if $f_n$ is given such that
$\|f_n(t)\|\leq C\|f(t_0)\|$ with $\|f(t_0)\|<\varepsilon_0$, then $g_{n+1}$ and $k_{n+1}$ are
constructed such that they satisfy the assumption {\sf (A)}. Now, Proposition \ref{Prop_Boltzmann} again
applies, we obtain $f_{n+1}$, and this completes the iteration.

We have obtained iteration functions $\{g_n\}$, $\{k_n\}$, and $\{f_n\}$.
The estimates of Section \ref{Sec_Einstein} show that
the following quantities are uniformly bounded:
\[
\Big|t^2H_n(t)-\frac23t\Big|,\quad
|t^2F_n(t)|,\quad
|(\bar{g}_n)_{ab}(t)|,\quad
|(\bar{g}_n)^{ab}(t)|,\quad
|t^2(\dot{\bar{g}}_n)_{ab}(t)|,\quad
|t^2(\dot{\bar{g}}_n)^{ab}(t)|,
\]
where $H_n$ and $F_n$ are understood as the Hubble variable and the shear constructed by $g_n$ and $k_n$.
Using the evolution equations \eqref{evolution1}--\eqref{evolution2}, we can see that
the following quantities are also bounded:
\[
|t^{-\frac13}(k_n)_{ab}(t)|,\quad
|t^{\frac23}(\dot{k}_n)_{ab}(t)|.
\]
The estimates of Section \ref{Sec_Boltzmann} show that the following expressions are bounded:
\[
\exp(t^{-\varepsilon_1}(\bar{g}_n)^{ab}p_ap_b)f_n(t,p_*),\quad
|\exp(t^{-\varepsilon_1}(\bar{g}_n)^{ab}p_ap_b)\partial_{p_k}f_n(t,p_*)|.
\]
Taking limit, up to a subsequence, we find continuous functions $g$, $k$, and $f$.
Using the equations \eqref{evolution1}--\eqref{evolution2} again, it turns out that
$g$ and $k$ are classical solutions of \eqref{evolution1}--\eqref{evolution2}.
Similarly, we can show that $f$ is a classical solution of the Boltzmann equation
with additional estimates.
Uniqueness is also easily proved by a standard argument, and asymptotic behaviours
are obtained by the arguments of Section \ref{Sec_Einstein} and \ref{Sec_Boltzmann}.
The following is the main theorem.

\begin{thm}\label{Thm}
Consider the Einstein-Boltzmann system with Bianchi I symmetry \eqref{evolution1}--\eqref{S_ab}.
Suppose that the assumption on the scattering kernel holds
and the Hubble variable is initially positive, and set the time origin as in \eqref{choice}.
Let $g_{ab}(t_0)$, $k_{ab}(t_0)$, and $f(t_0)$ be initial data of the Einstein-Boltzmann system
satisfying the constraints \eqref{constraint1}--\eqref{constraint2} such that
$f(t_0)$ is $C^1$ and $F(t_0)\neq 0$ is sufficiently small.
Then, there exist small positive constants $\varepsilon_0$
and $\varepsilon_1$ such that if $\|f(t_0)\|<\varepsilon_0$, then there exist
unique classical solutions $g_{ab}$, $k_{ab}$, and $f$ to the Einstein-Boltzmann system
corresponding to the initial data.
The solutions exist globally in time, and the distribution function $f$ is nonnegative.
Moreover, there exist a small positive constant $\varepsilon$ and constant matrices $\mathcal{G}_{ab}$
and $\mathcal{G}^{ab}$ such that the following estimates hold:
\begin{eqnarray*}
H(t)&=&\frac23t^{-1}(1+O(\varepsilon t^{-1})),\\
F(t)&=&O(\varepsilon t^{-2}),\\
g_{ab}(t)&=&t^{+\frac43}\Big(\mathcal{G}_{ab}+O(\varepsilon t^{-1})\Big),\\
g^{ab}(t)&=&t^{-\frac43}\Big(\mathcal{G}^{ab}+O(\varepsilon t^{-1})\Big).
\end{eqnarray*}
The distribution function $f$ is bounded as $\|f(t)\|\leq C\|f(t_0)\|$,
and in particular it satisfies
\[
\hat{f}(t,\hat{p})\leq C\|f(t_0)\|\exp(-t^{-\varepsilon_1+\frac43}|\hat{p}|^2)
\]
in an orthonormal frame.
\end{thm}

With the results of the theorem, we are able to show that the spacetime behaves as dust at late times.
We show that $\frac{S}{\rho}$ tends to zero which means that the ``pressure'' tends to zero
in comparison with the energy density. In an orthonormal frame the estimates of $S$ and $\rho$ are as follows:
\begin{align*}
\rho&=\int_{\bbr^3}\hat{f}(t,\hat{p})(1+|\hat{p}|^2)^{\frac12}d\hat{p}
\leq C\int_{\bbr^3}\exp(-t^{-\varepsilon_1+\frac43}|\hat{p}|^2)(1+|\hat{p}|^2)^{\frac12}d\hat{p}
\leq Ct^{-2+\frac32\varepsilon_1},\\
|\hat{S}_{ab}|&\leq\int_{\bbr^3}\hat{f}(t,\hat{p})|\hat{p}_a||\hat{p}_b|(1+|\hat{p}|^2)^{-\frac12}d\hat{p}
\leq C\int_{\bbr^3}\exp(-t^{-\varepsilon_1+\frac43}|\hat{p}|^2)|\hat{p}|^2d\hat{p}
\leq Ct^{-\frac{10}{3}+\frac{5}{2}\varepsilon_1}.
\end{align*}
Since $S=\eta^{ab}\hat{S}_{ab}$, we obtain that $S=O(t^{-\frac{10}{3}+\frac52\varepsilon_1})$
and $\rho=O(t^{-2+\frac32\varepsilon_1})$.
For $\frac{S}{\rho}$, we need an estimate of $\rho$ in the other direction.
We can use the fact that $\rho\geq N^0$. Using the upper bound of $H$ in \eqref{estH}
and the differential equation (\ref{N0}), we have
\begin{eqnarray*}
\dot{N}^0=-3HN^0\geq -2t^{-1}N^0,
\end{eqnarray*}
which can be integrated to obtain
\begin{eqnarray*}
\rho(t)\geq N^0(t)\geq N^0(t_0)\bigg(\frac{t}{t_0}\bigg)^{-2}.
\end{eqnarray*}
We obtain the following corollary.

\begin{cor}
Consider the solutions of the Einstein-Boltzmann system given in Theorem \ref{Thm}.
There exists a small positive constant $\varepsilon$ such that
the matter terms have the following asymptotic behaviour at late times:
\begin{eqnarray*}
\rho&=&O(t^{-2+\varepsilon}),\\
S&=&O(t^{-\frac{10}{3}+\varepsilon}),\\
\frac{S}{\rho}&=&O(t^{-\frac43+\varepsilon}).
\end{eqnarray*}
\end{cor}

Another consequence of the theorem is the following:
\begin{cor}
Consider the solutions of the Einstein-Boltzmann system given in Theorem \ref{Thm}. The spacetime in consideration is future geodesically complete.
\end{cor}
In order to prove this, the main point is to consider the relation between the time variable $t$ and the affine parameter $\tau$ along any future directed causal geodesic which is:
\begin{eqnarray*}
\frac{d\tau}{dt}=\frac{1}{p^0}=\frac{1}{\sqrt{1+g^{ij}p_i p_j}}>C.
\end{eqnarray*}
The last inequality is obtained by realizing that the $p_i$ are constant along the characteristics for Bianchi I spacetimes and that the inverse of the metric decays. Integrating the inequality we see that $\tau$ goes to infinity as $t$ does.

\section{Outlook}
We have shown under small data assumptions and an assumption on the scattering kernel, that solutions to the Einstein-Boltzmann system with Bianchi I symmetry are future geodesically complete, that they isotropize and have a dust-like behaviour at late times. Spacetimes of Bianchi type I are the simplest case of the Bianchi class A. It would be thus of interest to consider other Bianchi types. Here we have treated massive particles and a natural generalization would be to consider massless particles or to consider a collection of particles with different masses. Our assumption on the scattering kernel is a technical assumption and we plan to remove or at least improve this assumption in the future. In the present paper we were dealing in a homogeneous setting. It is a simplifying assumption which should be removed as well, since conclusions do not have to prevail in an inhomogeneous setting as showed in \cite{preHans}.

\section*{Appendix}
In this part we study the Boltzmann collision operator and parametrization
of post-collision momenta in relativistic cases.
We first consider the Minkowski case, in which case
two different representations of the collision operator are known,
which are given by Glassey and Strauss \cite{GS93} and Strain \cite{S11}.
We use the argument given in \cite{GS93} to derive a similar
form of the collision operator, and it will be shown that
this form of the collision operator can also be derived in a direct way.
To extend these representations to the Bianchi type I case we apply the orthonormal
frame approach. An explicit expression of orthonormal frame will be given, and
then the representations for the Bianchi type I case will be obtained.

\subsection*{Collision operator in the Minkowski case}
Let us consider the collision operator of the
Boltzmann equation in the special relativistic case.
In its original form we have
\[
Q(f,f)=\frac{1}{p^0}\iiint_{\bbr^9}W(p,q|p',q')
\Big(f(p')f(q')-f(p)f(q)\Big)\frac{dp'}{p'^0}\frac{dq'}{q'^0}\frac{dq}{q^0},
\]
where the quantity $W$ is called
the transition rate \cite{S11}, and it is given by
\[
W(p,q|p',q')=\frac{s}{2}\sigma(h,\theta)\delta^{(4)}(p^\mu+q^\mu - p'^\mu - q'^\mu),
\]
where $\delta^{(4)}$ is the four-dimensional Dirac delta function which expresses the conservation of energy and momentum.
It is known in \cite{S11} that the nine dimensional integration above is reduced to
\begin{equation}
Q(f,f)=
\int_{\bbr^3}\int_{\bbs^2}\frac{h\sqrt{s}}{4p^0q^0}\sigma(h,\theta)
\Big(f(p')f(q')-f(p)f(q)\Big)
d\omega dq,\label{appendix1}
\end{equation}
and the post-collision momentum is parametrized as follows:
\begin{equation}
\left(
\begin{array}{c}
p'^0\\
p'^k
\end{array}
\right)=
\left(
\begin{array}{c}
\displaystyle
\frac{p^0+q^0}{2}
+\frac{h}{2}
\frac{(n\cdot\omega)}{\sqrt{s}}\\
\displaystyle
\frac{p^k+q^k}{2}
+\frac{h}{2}
\bigg(\omega^k
+\bigg(\frac{n^0}{\sqrt{s}}-1\bigg)\frac{(n\cdot\omega)n^k}{|n|^2}\bigg)
\end{array}
\right),\label{appendix2}
\end{equation}
where $\omega\in\bbs^2$.
To obtain the above representations, a suitable Lorentzian transformation
should be considered.

On the other hand, the nine dimensional
integration can be reduced in a different way
without using Lorentz transformations so that a different form of the collision operator is derived.
The post-collision momentum may be assumed to have the form of
$p'=p+r\xi$ as in \cite{GS93}, but we replace this by
\begin{equation}
p'=\frac{p+q}{2}+r\xi,\quad\xi\in\bbs^2,\label{appendix3}
\end{equation}
and then follow the calculations of \cite{GS93}.
We first consider the integration over $\bbr^6$ with respect to $p'$ and $q'$:
\begin{align*}
&\iint_{\bbr^6}\sigma(h,\theta)\delta^{(4)}(p^\mu+q^\mu - p'^\mu - q'^\mu)
\Big(f(p')f(q')-f(p)f(q)\Big)\frac{dp'}{p'^0}\frac{dq'}{q'^0}\\
&=\int_{\bbr^3}\sigma(h,\theta)\delta(p^0+q^0 - p'^0 - q'^0)
\Big(f(p')f(q')-f(p)f(q)\Big)\frac{dp'}{p'^0q'^0}\\
&=\int_0^\infty\int_{\bbs^2}\sigma(h,\theta)\delta(p^0+q^0 - p'^0 - q'^0)
\Big(f(p')f(q')-f(p)f(q)\Big)\frac{r^2d\xi dr}{p'^0q'^0},
\end{align*}
where $q'=p+q-p'$ in the first identity, and we assume \eqref{appendix3}
in the second identity. For the Dirac delta function we use the property
$\delta(x)=|a|\delta(ax)$ to obtain
\begin{align*}
&\delta(p^0+q^0 - p'^0- q'^0)\\
&=(p^0+q^0+p'^0+q'^0)\delta\Big((p^0+q^0)^2-(p'^0+q'^0)^2\Big)\\
&=2(p^0+q^0)\delta\Big((p^0+q^0)^2-(p'^0+q'^0)^2\Big).
\end{align*}
Similar arguments give the following calculations:
\begin{align*}
&2(p^0+q^0)\delta\Big((p'^0+q'^0)^2-(p^0+q^0)^2\Big)\\
&=2(p^0+q^0)\delta\Big(2p'^0q'^0-(p^0+q^0)^2+(p'^0)^2+(q'^0)^2\Big)\\
&=8(p^0+q^0)p'^0q'^0\delta\Big(4(p'^0)^2(q'^0)^2-((p^0+q^0)^2-(p'^0)^2-(q'^0)^2)^2\Big),
\end{align*}
and let $p=p(r)$ denote the quantity contained in the Dirac delta function.
Then, it is rewritten as follows:
\begin{align*}
p(r)
&=4(p'^0)^2(q'^0)^2-(p^0+q^0)^4
+2(p^0+q^0)^2((p'^0)^2+(q'^0)^2)
-((p'^0)^2+(q'^0)^2)^2\\
&=-(p^0+q^0)^4-((p'^0)^2-(q'^0)^2)^2
+2(p^0+q^0)^2((p'^0)^2+(q'^0)^2).
\end{align*}
Note that
\begin{align*}
((p'^0)^2-(q'^0)^2)^2
&=(|p'|^2-|q'|^2)^2
=\bigg(\bigg|\frac{p+q}{2}+r\xi\bigg|^2-\bigg|
\frac{p+q}{2}-r\xi\bigg|^2\bigg)^2\\
&=(2(p+q)\cdot r\xi)^2
=4(n\cdot\xi)^2r^2,
\end{align*}
and
\begin{align*}
(p'^0)^2+(q'^0)^2
&=2+\bigg|\frac{p+q}{2}+r\xi\bigg|^2+\bigg|
\frac{p+q}{2}-r\xi\bigg|^2
=2+\frac{|n|^2}{2}+2r^2,
\end{align*}
where $n=p+q$ and $n^0=p^0+q^0$.
Therefore, the quantity $p(r)$ is written as
\begin{align*}
p(r)&=-(n^0)^4-((p'^0)^2-(q'^0)^2)^2
+2(n^0)^2((p'^0)^2+(q'^0)^2)\\
&=-(n^0)^4-4(n\cdot\xi)^2r^2
+2(n^0)^2\Big(2+\frac{|n|^2}{2}+2r^2\Big)\\
&=r^2(4(n^0)^2-4(n\cdot\xi)^2)
-(n^0)^2((n^0)^2-|n|^2-4).
\end{align*}
We may write $(n^0)^2-(n\cdot\xi)^2=t_\alpha t^\alpha$ for $t^\alpha=(n\cdot\xi,n^0\xi)$
to obtain
\[
p(r)=4(t_\alpha t^\alpha) r^2-h^2(n^0)^2,
\]
where we used $(n^0)^2-|n|^2-4=s-4=h^2$.
We plug this to the above calculation:
\begin{align*}
&\int_0^\infty\int_{\bbs^2}\sigma(h,\theta)
\delta(p^0+q^0 - p'^0- q'^0)
\Big(f(p')f(q')-f(p)f(q)\Big)\frac{r^2d\xi dr}{p'^0q'^0}\\
&=8\int_0^\infty\int_{\bbs^2}\sigma(h,\theta)
(p^0+q^0)p'^0q'^0\delta(p(r))
\Big(f(p')f(q')-f(p)f(q)\Big)\frac{r^2d\xi dr}{p'^0q'^0}\\
&=8\int_0^\infty\int_{\bbs^2}\sigma(h,\theta)
n^0\delta\Big(4(t_\alpha t^\alpha)r^2-h^2(n^0)^2\Big)
\Big(f(p')f(q')-f(p)f(q)\Big)r^2 d\xi dr\allowdisplaybreaks\\
&=8\int_0^\infty\int_{\bbs^2}\sigma(h,\theta)
\frac{n^0}{4t_\alpha t^\alpha}\delta\bigg(r^2-\frac{h^2(n^0)^2}{4t_\alpha t^\alpha}\bigg)
\Big(f(p')f(q')-f(p)f(q)\Big)r^2 d\xi dr\\
&=8\int_0^\infty\int_{\bbs^2}\sigma(h,\theta)
\frac{n^0}{4t_\alpha t^\alpha}\frac{1}{2r}\delta\bigg(r-\frac{hn^0}{2\sqrt{t_\alpha t^\alpha}}\bigg)
\Big(f(p')f(q')-f(p)f(q)\Big)r^2 d\xi dr\\
&=\int_0^\infty\int_{\bbs^2}\sigma(h,\theta)
\frac{n^0}{t_\alpha t^\alpha}\delta\bigg(r-\frac{hn^0}{2\sqrt{t_\alpha t^\alpha}}\bigg)
\Big(f(p')f(q')-f(p)f(q)\Big)r d\xi dr\allowdisplaybreaks\\
&=\int_{\bbs^2}\sigma(h,\theta)
\frac{n^0}{t_\alpha t^\alpha}\frac{hn^0}{2\sqrt{t_\alpha t^\alpha}}
\Big(f(p')f(q')-f(p)f(q)\Big) d\xi\\
&=\frac{1}{2}\int_{\bbs^2}\sigma(h,\theta)
\frac{h(n^0)^2}{(t_\alpha t^\alpha)^{3/2}}
\Big(f(p')f(q')-f(p)f(q)\Big)d\xi.
\end{align*}
Since $t_\alpha t^\alpha=(n^0)^2-(n\cdot\xi)^2$,
we combine the other quantities to get
\begin{equation}
Q(f,f)=\int_{\bbr^3}\int_{\bbs^2}
\frac{hs(n^0)^2\sigma(h,\theta)}{4p^0q^0((n^0)^2-(n\cdot\xi)^2)^{3/2}}
\Big(f(p')f(q')-f(p)f(q)\Big)d\xi dq.\label{appendix4}
\end{equation}
From the above calculations we find $r=\frac{hn^0}{2\sqrt{t_\alpha t^\alpha}}$ to write
the post-collision momentum as
\begin{equation}
\left(
\begin{array}{c}
p'^0\\
p'^k
\end{array}
\right)=
\left(
\begin{array}{c}
\displaystyle
\frac{p^0+q^0}{2}+\frac{h}{2}\frac{(n\cdot\xi)}{\sqrt{(n^0)^2-(n\cdot\xi)^2}}\\
\displaystyle
\frac{p^k+q^k}{2}+\frac{h}{2}\frac{n^0\xi^k}{\sqrt{(n^0)^2-(n\cdot\xi)^2}}
\end{array}
\right),
\label{appendix5}
\end{equation}
where $\xi\in\bbs^2$.

We may now compare \eqref{appendix1} and \eqref{appendix4} to see that the two
parameters $\omega$ and ${\xi}$ have the following relation:
\begin{equation}
d\omega=\frac{\sqrt{s}(n^0)^2}{((n^0)^2-(n\cdot\xi)^2)^{3/2}}d\xi,\label{appendix6}
\end{equation}
and this can be proved by direct calculations. Suppose that a post-collision
momentum $p'^\alpha$ is given. It has two different representations \eqref{appendix2}
and \eqref{appendix5}, and we have the following identities:
\begin{align*}
\frac{n\cdot\omega}{\sqrt{s}}&=\frac{n\cdot\xi}{\sqrt{(n^0)^2-(n\cdot\xi)^2}},\\
\omega^k+\bigg(\frac{n^0}{\sqrt{s}}-1\bigg)
\frac{(n\cdot \omega)n^k}{|n|^2}
&=\frac{n^0\xi^k}{\sqrt{(n^0)^2-(n\cdot\xi)^2}}.
\end{align*}
Taking coordinates satisfying $n=(0,0,|n|)$, we write
$\omega=\omega(\alpha,\beta)$ and $\xi=\xi(\phi,\theta)$ as follows:
\begin{align*}
\omega&=(\sin\alpha\cos\beta,\sin\alpha\sin\beta,\cos\alpha),\\
\xi&=(\sin\phi\cos\theta,\sin\phi\sin\theta,\cos\phi),
\end{align*}
and then the second identity reduces to
\[
\Big(\sin\alpha\cos\beta,\sin\alpha\sin\beta,\frac{n^0}{\sqrt{s}}\cos\alpha\Big)
=\frac{n^0(\sin\phi\cos\theta,\sin\phi\sin\theta,\cos\phi)}{\sqrt{(n^0)^2-|n|^2\cos^2\phi}},
\]
where $\alpha,\phi\in[0,\pi]$ and $\beta,\theta\in[0,2\pi]$.
This quantity describes a certain two-dimensional surface in $\bbr^3$,
and the above identity shows that the surface is parametrized in two different ways.
This implies the surface element can be computed in two different ways.
Let $X$ denote the surface. Then, we have
\[
dS=|\partial_\alpha X\times\partial_\beta X|d\alpha d\beta
=|\partial_\phi X\times\partial_\theta X|d\phi d\theta.
\]
By direct calculations, we have
\begin{align*}
\partial_\alpha X&=\Big(\cos\alpha\cos\beta,
\cos\alpha\sin\beta,-\frac{n^0}{\sqrt{s}}\sin\alpha\Big),\\
\partial_\beta X&=\Big(
-\sin\alpha\sin\beta,\sin\alpha\cos\beta,0\Big),
\end{align*}
and
\[
\partial_\alpha X\times\partial_\beta X=
\Big(\frac{n^0}{\sqrt{s}}\sin^2\alpha\cos\beta,
\frac{n^0}{\sqrt{s}}\sin^2\alpha\sin\beta,
\sin\alpha\cos\alpha\Big).
\]
Hence, we have
\begin{align*}
|\partial_\alpha X\times\partial_\beta X|^2&=
\frac{(n^0)^2}{s}\sin^4\alpha
+\sin^2\alpha\cos^2\alpha
=\frac{\sin^2\alpha}{s}\Big((n^0)^2\sin^2\alpha+s\cos^2\alpha\Big)\\
&=\frac{\sin^2\alpha}{s}\Big((n^0)^2-|n|^2\cos^2\alpha\Big)
=\frac{\sin^2\alpha}{s}\Big((n^0)^2-(n\cdot\omega)^2\Big),
\end{align*}
and therefore we obtain
\begin{equation}
|\partial_\alpha X\times\partial_\beta X|d\alpha d\beta
=\frac{\sqrt{(n^0)^2-(n\cdot\omega)^2}}{\sqrt{s}}d\omega,\label{appendix7}
\end{equation}
where we used $d\omega=\sin\alpha d\alpha d\beta$.
We now consider the second case $X=X(\phi,\theta)$.
By direct calculations again, we have
\begin{align*}
\partial_\phi X&=\partial_\phi\bigg[\frac{n^0}{\sqrt{(n^0)^2-|n|^2\cos^2\phi}}\bigg]
(\sin\phi\cos\theta,\sin\phi\sin\theta,\cos\phi)\\
&\quad+\frac{n^0}{\sqrt{(n^0)^2-|n|^2\cos^2\phi}}
(\cos\phi\cos\theta,\cos\phi\sin\theta,-\sin\phi)\allowdisplaybreaks\\
&=\frac{-n^0|n|^2\sin\phi\cos\phi}{((n^0)^2-|n|^2\cos^2\phi)^{3/2}}
(\sin\phi\cos\theta,\sin\phi\sin\theta,\cos\phi)\\
&\quad+\frac{n^0}{\sqrt{(n^0)^2-|n|^2\cos^2\phi}}
(\cos\phi\cos\theta,\cos\phi\sin\theta,-\sin\phi),\\
\intertext{and}
\partial_\theta X&=\frac{n^0}{\sqrt{(n^0)^2-|n|^2\cos^2\phi}}
(-\sin\phi\sin\theta,\sin\phi\cos\theta,0).
\end{align*}
Note that $|\partial_\phi X\times\partial_\theta X|^2=|\partial_\phi X|^2|\partial_\theta X|^2$,
where the quantities are calculated as
\begin{align*}
|\partial_\phi X|^2&=
\frac{(n^0)^2|n|^4\sin^2\phi\cos^2\phi}{((n^0)^2-|n|^2\cos^2\phi)^3}
+\frac{(n^0)^2}{(n^0)^2-|n|^2\cos^2\phi},\\
\intertext{and}
|\partial_\theta X|^2&=
\frac{(n^0)^2\sin^2\phi}{(n^0)^2-|n|^2\cos^2\phi}.
\end{align*}
Hence, we have
\begin{align*}
&|\partial_\phi X\times \partial_\theta X|^2\\
&=\frac{(n^0)^4\sin^2\phi}{((n^0)^2-|n|^2\cos^2\phi)^4}
\Big(|n|^4\sin^2\phi\cos^2\phi+((n^0)^2-|n|^2\cos^2\phi)^2\Big)\\
&=\frac{(n^0)^4\sin^2\phi}{((n^0)^2-|n|^2\cos^2\phi)^4}
\Big(|n|^4\cos^2\phi+(n^0)^4
-2(n^0)^2|n|^2\cos^2\phi\Big)\allowdisplaybreaks\\
&=\frac{(n^0)^4\sin^2\phi}{((n^0)^2-(n\cdot\xi)^2)^4}
\Big(|n|^2(n\cdot\xi)^2+(n^0)^4
-2(n^0)^2(n\cdot\xi)^2\Big)\\
&=\frac{(n^0)^4\sin^2\phi}{((n^0)^2-(n\cdot\xi)^2)^4}
\Big((n^0)^4-(n^0)^2(n\cdot\xi)^2-s(n\cdot\xi)^2\Big)\\
&=\frac{(n^0)^4\sin^2\phi}{((n^0)^2-(n\cdot\xi)^2)^3}
\bigg((n^0)^2-\frac{s(n\cdot\xi)^2}{(n^0)^2-(n\cdot\xi)^2}\bigg).
\end{align*}
On the other hand, we can use the first relation between $\xi$ and $\omega$
to the last quantity as follows:
\[
(n^0)^2-\frac{s(n\cdot\xi)^2}{(n^0)^2-(n\cdot\xi)^2}
=(n^0)^2-(n\cdot\omega)^2.
\]
Then, we obtain
\begin{equation}
|\partial_\phi X\times\partial_\theta X|d\phi d\theta
=\frac{(n^0)^2\sqrt{(n^0)^2-(n\cdot\omega)^2}}
{((n^0)^2-(n\cdot\xi)^2)^{3/2}}d\xi,\label{appendix8}
\end{equation}
where $d\xi=\sin\phi d\phi d\theta$.
Finally, we know from calculus that the surface element
does not depend on the way in which the surface is parametrized,
and therefore we compare \eqref{appendix7} and \eqref{appendix8}
to get the desired result:
\[
d\omega=\frac{\sqrt{s}(n^0)^2}{((n^0)^2-(n\cdot\xi)^2)^{3/2}}d\xi,
\]
which we have already observed in \eqref{appendix6}.

\section*{Acknowledgements}
H.\ Lee has been supported by the TJ Park Science Fellowship of POSCO TJ Park Foundation. This work was initiated while E.N. was funded by the Irish Research Council and he is currently funded by a Juan de la Cierva grant of the Spanish government. He thanks Kyung Hee University for the hospitality during his stay where part of this work was done and thanks  Alan D. Rendall and John Stalker for several discussions.

\end{document}